\documentclass[11pt]{article}
\usepackage{odonnell}
\usepackage{etoolbox}

\newcommand{\pE}{\widetilde{\E}}

\newcommand{\Bin}{\mathrm{Bin}}
\newcommand{\Ber}{\mathrm{Ber}}
\newcommand{\OBJ}{\mathrm{OBJ}}

\newcommand{\cI}{\mathcal{I}}
\newcommand{\Ind}{\mathbb{I}}

\newcommand{\cW}{\mathcal{W}}

\newcommand{\maxcut}{{\sc max-cut}\xspace}
\newcommand{\minbisection}{{\sc min-bisection}\xspace}
\newcommand{\ER}{Erd\H{o}s--R\'enyi\xspace}
\newcommand{\txor}{{\sc 2-xor}\xspace}
\newcommand{\xor}[1]{{\sc {#1}-xor}\xspace}
\newcommand{\sat}[1]{{\sc {#1}-sat}\xspace}
\newcommand{\ippi}[2]{\langle #1, #2\rangle_{\pi}}

\newcommand{\X}{\mathsf{X}}
\newcommand{\Y}{\mathsf{Y}}
\newcommand{\ZZ}{\mathsf{Z}}
\newcommand{\psdone}{\widetilde{\Psi}}
\newcommand{\psdtwo}{\Psi}
\newcommand{\specrad}{\rho}

\begin{document}

\title{Sherali--Adams Strikes Back}
\author{Ryan O'Donnell\thanks{Computer Science Department, Carnegie Mellon University.  \texttt{odonnell@cs.cmu.edu}. Supported by NSF grants CCF-1618679, CCF-1717606. This material is based upon work supported by the National Science Foundation under grant numbers listed above. Any opinions, findings and conclusions or recommendations expressed in this material are those of the author and do not necessarily reflect the views of the National Science Foundation (NSF).} \and Tselil Schramm\thanks{Harvard University and MIT. \texttt{tselil@seas.harvard.edu}. This work was supported by NSF grants CCF 1565264 and CNS 1618026}. }

\maketitle

\begin{abstract}
    Let $G$ be any $n$-vertex graph whose random walk matrix has its nontrivial eigenvalues bounded in magnitude by $1/\sqrt{\Delta}$ (for example, a random graph $G$ of average degree~$\Theta(\Delta)$ typically has this property).  We show that the $\exp\parens*{c \frac{\log n}{\log \Delta}}$-round Sherali--Adams linear programming hierarchy certifies that the maximum cut in such a~$G$ is at most $50.1\%$ (in fact, at most $\tfrac12 + 2^{-\Omega(c)}$).  For example, in random graphs with $n^{1.01}$ edges, $O(1)$ rounds suffice; in random graphs with $n \cdot \polylog(n)$ edges, $n^{O(1/\log \log n)} = n^{o(1)}$ rounds suffice.  
    
    Our results stand in contrast to the conventional beliefs that linear programming hierarchies perform poorly for \maxcut and other CSPs, and that eigenvalue/SDP methods are needed for effective refutation.  Indeed, our results imply that constant-round Sherali--Adams can strongly refute random Boolean $k$-CSP instances with $n^{\lceil k/2 \rceil + \delta}$ constraints; previously this had only been done with spectral algorithms or the SOS SDP hierarchy.
\end{abstract}
\thispagestyle{empty}
\clearpage

\setcounter{page}{1}
\section{Introduction}

Linear programming (LP) is a fundamental algorithmic primitive, and is the method of choice for a huge number of optimization and approximation problems.  Still, there are some very basic tasks where it performs poorly.  A classic example is the simplest of all constraint satisfaction problems (CSPs), the \maxcut problem: Given a graph $G = (V,E)$, partition $V$ into two parts so as to maximize the fraction of `cut' (crossing) edges.  The standard LP relaxation for this problem~\cite{BM86,Pol92} involves optimizing over the \emph{metric polytope}. Using ``$\pm 1$ notation'', we have a variable $\Y_{uv}$ for each pair of vertices $\{u,v\}$ (with $\Y_{uv}$ supposed to be $-1$ if the edge is cut, $+1$ otherwise); the LP is:
\begin{align*}
    &&  -1 \leq \Y_{uv} \leq 1 & \quad \text{(for all $u,v \in V$)}\\
    \text{\maxcut}(G) \ \leq \ \max\ \smash{\frac12 - \frac12 \cdot \frac{1}{|E|}\sum_{uv \in E} \Y_{uv}}\quad\quad \text{s.t.}&&-\Y_{uv} - \Y_{vw} - \Y_{wu} \leq 1 &\quad  \text{(for all $u,v,w \in V$)}\\
    &&-\Y_{uv} + \Y_{vw} + \Y_{wu} \leq 1 &\quad \text{(for all $u,v,w \in V$)}
\end{align*}
While this LP gives the optimal bound for  some graphs (precisely, all graphs not contractible to~$K_5$~\cite{BM86}), it can give a very poor bound in general.  Indeed, although there are graphs with maximum cut arbitrarily close to $1/2$ (e.g.,~$K_n$), the above LP bound is at least $2/3$ for every graph, since $\Y_{uv} \equiv -1/3$ is always a valid solution.  Worse, there are graphs $G$ with $\text{\maxcut}(G)$ arbitrarily close to~$1/2$ but with LP value arbitrarily close to~$1$ --- i.e., graphs where the \emph{integrality ratio} is $2-o(1)$.  For example, this is true~\cite{PT94} of an \ER $\calG(n,\Delta/n)$ random graph with high probability (whp) when $\Delta = \Delta(n)$ satisfies $\omega(1) < \Delta < n^{o(1)}$.

There have been two main strategies employed for overcoming this deficiency: strengthened LPs, and eigenvalue methods.

\paragraph{Strengthened LPs.} One way to try to improve the performance of LPs on \maxcut is to add more valid inequalities to the LP relaxation, beyond just the ``triangle inequalities''.  Innumerable valid inequalities have been considered: $(2k+1)$-gonal, hypermetric, negative type, gap, clique-web, suspended tree, as well as inequalities from the Lov\'asz--Schrijver hierarchy;  see Deza and Laurent~\cite[Ch.~28--30]{DL97} for a review.

It is now known that the most principled and general form of this strategy is the \emph{Sherali--Adams LP hierarchy}~\cite{SA90}, reviewed in \Cref{sec:proof-systems}.
At a high level, the Sherali--Adams LP hierarchy gives a standardized way to tighten LP relaxations of Boolean integer programs, by adding variables and constraints.
The number of new variables/constraints is parameterized by a positive integer $R$, called the number of ``rounds''.  Given a Boolean optimization problem with $n$ variables, the $R$-round Sherali--Adams LP has variables and constraints corresponding to monomials of degree up to~$R$, and thus has size $O(n)^{R}$.
A remarkable recent line of work~\cite{CLRS16,KothariMR17} has shown that for any CSP (such as \maxcut), the $R$-round Sherali--Adams LP relaxation achieves essentially the tightest integrality ratio among all LPs of its size. 
Nevertheless, even this most powerful of LPs arguably struggles to certify good bounds for \maxcut.
In a line of work \cite{dlVM07,SchoenebeckTT07} concluding in a result of Charikar--Makarychev--Makarychev~\cite{CMM09}, it was demonstrated that for any constant $\eps > 0$, there are graphs (random $\Delta$-regular ones, $\Delta = O(1)$) for which the $n^{\Omega(1)}$-round Sherali--Adams LP has a \maxcut integrality gap of $2 - \eps$.  As a consequence, \emph{every} \maxcut LP relaxation of size up to $2^{n^{\Omega(1)}}$ has such an integrality gap.

\paragraph{Eigenvalue and SDP methods.}  But for \maxcut, there is a simple, non-LP, algorithm that works very well to certify that random graphs have maximum cut close to~$1/2$: the \emph{eigenvalue} bounds.  There are two slight variants here (that coincide in the case of regular graphs):  Given graph $G = (V,E)$ with adjacency matrix~$A$ and diagonal degree matrix~$D$, the eigenvalue bounds are
\begin{align}
    \text{\maxcut}(G) &\leq \frac{|V|}{4|E|} \lambda_{\text{max}}(D - A) \label[ineq]{eqn:eigval1}\\
    \text{\maxcut}(G) &\leq \frac12 +\frac12 \lambda_{\text{max}}(-D^{-1}A). \label[ineq]{eqn:eigval2}
\end{align}
Here $D-A$ and $D^{-1}A$ are the \emph{Laplacian} matrix and the \emph{random walk} matrix, respectively. The use of eigenvalues to bound various cut values in graphs (problems like \maxcut, \minbisection, \txor, expansion, etc.)\ has a long history dating back to Fieldler and Donath--Hoffman~\cite{fiedler1973algebraic,donath2003lower} among others (\Cref{eqn:eigval1} is specifically from Mohar--Poljak~\cite{mohar1993eigenvalues}).  It was recognized early on that eigenvalue methods work particularly well for solving planted-random instances (e.g., of \txor~\cite{Hastad84} and \minbisection~\cite{Boppana87}) and for certifying \maxcut values near $1/2$ for truly random instances.  Indeed, as soon as one knows (as we now do~\cite{TY16,FO05}) that $D^{-1}A$ has all nontrivial eigenvalues bounded in magnitude by $O(1/\sqrt{\Delta})$  (whp) for a random $\Delta$-regular graph (or an \ER $\mathcal{G}(n,\Delta/n)$ graph with $\Delta \gtrsim \log n$), the eigenvalue bound \Cref{eqn:eigval2} certifies that $\text{\maxcut}(G) \leq 1/2 + O(1/\sqrt{\Delta})$.  This implies an integrality ratio tending to~$1$;  indeed, $\text{\maxcut}(G) = 1/2 + \Theta(1/\sqrt{\Delta})$ in such random graphs (whp).

Furthermore, if one extends the  eigenvalue bound \Cref{eqn:eigval1} above to
\[
    \text{\maxcut}(G) \leq \min_{\substack{U \text{ diagonal} \\ \mathrm{tr}(U) = 0}} \frac{|V|}{4|E|} \lambda_{\text{max}}(D - A + U)
\]
(as suggested by Delorme and Poljak \cite{DP93}, following \cite{donath2003lower,Boppana87}), one obtains the polynomial-time computable \emph{semidefinite programming (SDP)} bound.  Goemans and Williamson \cite{GoemansW95} showed this bound has integrality ratio less than $1.14 \approx 1/.88$ for \emph{worst-case}~$G$, and it was subsequently shown~\cite{Zwick1999, feige2001rpr, charikar2004maximizing} that the SDP bound is $1/2 + o(1)$ whenever $\text{\maxcut}(G) \leq 1/2 + o(1)$.

\paragraph{LPs cannot compete with eigenvalues/SDPs?}  This seemingly striking separation between the performance of LPs and SDPs  in the context of random \maxcut instances is now taken as a matter of course.  To quote, e.g., \cite{Trevisan09},
 \begin{quotation}
    [E]xcept for semidefinite programming, we know of no technique that can provide, for every graph of max cut optimum $\leq .501$, a certificate that its optimum is $\leq .99$. Indeed, the results of \cite{dlVM07,SchoenebeckTT07}[\cite{CMM09}] show that large classes of Linear Programming relaxations of max cut are unable to distinguish such instances.
\end{quotation}
Specifically, the last statement here is true for $\Delta$-regular random graphs when $\Delta$ is a certain large constant. 
The conventional wisdom is that for such graphs, linear programs cannot compete with semidefinite programs, and cannot certify even the eigenvalue bound.\\

Our main result challenges this conception.

\subsection{Our results}
We show that whenever the eigenvalue bound \Cref{eqn:eigval2} certifies the bound ${\text \maxcut}(G) \leq 1/2 + o(1)$, then $n^{o(1)}$-round Sherali--Adams can certify this as well.\footnote{Actually, there is a slight mismatch between our result and \Cref{eqn:eigval2}: in \Cref{thm:main} we need the maximum eigenvalue \emph{in magnitude} to be small; i.e., we need $\lambda_{\text{min}}(-D^{-1}A)$ to be not too negative.  This may well just be an artifact of our proof.}
\begin{theorem}[Simplified version of \Cref{thm:main-2xor} and \Cref{cor:main-2xor}]\label{thm:main}
Let $G$ be a simple $n$-vertex graph and assume that $|\lambda| < \rho$ for all eigenvalues $\lambda$ of $G$'s random walk matrix $D^{-1}A$ (excluding the trivial eigenvalue of $1$). 
Then for any $1 \leq c \leq \Omega(\log(1/\rho))$, Sherali--Adams with $n^{O(c/\log(1/\rho))}$ rounds certifies that ${\text \maxcut}(G) \leq 1/2 + 2^{-c}$.
\end{theorem}

For example, if $G$'s random walk matrix has its nontrivial eigenvalues bounded in magnitude by $n^{-.001}$, as is the case (whp) for random graphs with about $n^{1.002}$ edges, then Sherali--Adams can certify ${\text \maxcut}(G) \leq 50.1\%$ with \emph{constantly} many rounds.  We find this result surprising, and in defiance of the common belief that polynomial-sized LPs cannot take advantage of spectral properties of the underlying graph.\footnote{We should emphasize that it is not concomitant $n^{O(1)}$ running time that is surprising, since the eigenvalue bound itself already achieves that.}

(As an aside, in \Cref{prop:lb}, we show that the $R$-round Sherali--Adams relaxation for \maxcut has value at least $1/2 + \Omega(1/R)$ for \emph{every} graph~$G$.  This at least demonstrates that some dependence of the refutation strength on the number of Sherali--Adams rounds is always necessary.)

One might ask whether \Cref{thm:main} even requires the assumption of small eigenvalues.  That is, perhaps $n^{o(1)}$-round Sherali--Adams can certify $\text{\maxcut} \leq 1/2 + o(1)$ whenever this is \emph{true}.  As far as we know, this may be possible; after all, the basic SDP relaxation has this property~\cite{Zwick1999, feige2001rpr, charikar2004maximizing}.  On the other hand, the eigenvalue bound itself does \emph{not} have this property; there exist graphs with large (nontrivial) eigenvalues even though the maximum cut is close to~$1/2$.\footnote{Consider, for example, a graph given by the union of a $\Delta$-regular random graph on $n$ vertices and a $\Delta$-regular bipartite graph on $\sqrt{n}$ vertices. This will have \maxcut value close to $1/2$, but will also have large negative eigenvalues coming from the  bipartite component.}

\subsubsection{Subexponential-sized LPs for \maxcut in sparse random graphs}
One setting in which the spectral radius $\rho$ is understood concretely is in random regular graphs.
Building upon~\cite{FKS89,BFSU99,CGJ18}, the following was recently shown:
\begin{theorem*}[\cite{TY16}]\label{thm:TY} There is a fixed constant $C$ such that for all $3 \leq \Delta \leq n/2$ with $\Delta n$ even, it holds that a uniformly random $n$-vertex $\Delta$-regular simple graph $\bG$ satisfies the following with high probability: all eigenvalues of $\bG$'s normalized adjacency matrix, other than~$1$, are at most $C/\sqrt{\Delta}$ in magnitude.
\end{theorem*}
Combining the above with \Cref{thm:main}, we have the following consequence for \maxcut on random regular graphs:
\begin{corollary}
Let $n$, $3 \le \Delta\le n/2$, and $1 \le c \le \Omega(\log \Delta)$ be positive integers.
Then if $G$ is a random $\Delta$-regular $n$-vertex graph, with high probability $n^{O(c/\log \Delta)}$-round Sherali--Adams can certify that ${\text \maxcut}(G) \le \frac{1}{2} + 2^{-c}$.

\end{corollary}

For example, if $\Delta \ge C \cdot 10^6$ (for $C$ the constant in the bound on $\lambda(G)$), then $n^{1/3}$-rounds of Sherali--Adams can certify $\text{\maxcut}(G) \le .51$.
This result serves as a partial converse to \cite{CMM09}:
\begin{theorem*}(\cite[Theorem 5.3]{CMM09})
For every fixed integer $\Delta \ge 3$, with high probability over the choice of an $n$-vertex $\Delta$-regular random graph $G$,\footnote{In \cite{CMM09}, the graph is actually a {\em pruned} random graph, in which $o(n)$ edges are removed; this does not affect compatibility with our results, as the LP value is Lipschitz and so the pruning changes the LP value by $o(1)$.} 
the $n^{\Theta(1/f(\Delta))}$-round Sherali--Adams relaxation for \maxcut has value at least ${\text \maxcut}(G) \ge 0.99$, where $f(\Delta)$ is a function that grows with $\Delta$.\footnote{Though $f(\Delta)$ is not specified in \cite{CMM09} (and in their proof $f(\Delta)$ is dictated by a combination of lemmas in prior works) it appears we can take $f(\Delta) = \log \Delta$.}
\end{theorem*}

While \cite{CMM09} show that $\Delta$-regular random graphs require Sherali--Adams (and by \cite{KothariMR17} any LP) relaxations of at least subexponential size, our result implies that subexponential LPs are sufficient. 
Further, though the function $f(\Delta)$ is not specified in \cite{CMM09}, by tracing back through citations (e.g. \cite{AroraBLT06,AroraLNRRV12,CharikarMM10}) to extract a dependence, it appears we may take $f(\Delta) = \log\Delta$. 
So our upper bound is {\em tight} as a function of $\Delta$, up to constant factors.

Prior to our result, it was unclear whether even $(n/\polylog n)$-round Sherali--Adams could certify that the \maxcut value was $< 0.9$ for sparse random regular graphs.
Indeed, it was equally if not more conceivable that Charikar et al.'s result was not tight, and could be extended to $\tilde\Omega(n)$-rounds.
In light of our result, we are left to wonder whether there are instances of \maxcut which have truly exponential extension complexity.

\subsubsection{Refuting Random CSPs with linear programs}
With minor modifications, our argument extends as well to \txor.
Then following the framework in \cite{AOW15}, we have the following consequence for certifying bounds on the value of random $k$-CSPs:
\begin{theorem}[Simplified version of \Cref{thm:csps}]
Suppose that $P:\{\pm 1\}^k \to \{0,1\}$ is a $k$-ary Boolean predicate, and that $\delta,\epsilon > 0$.
Let $\E[P]$ be the probability that a random $x \in \{\pm 1\}^k$ satisfies $P$.
Then for a random instance $\cI$ of $P$ on $n$ variables with $m \ge n^{\lceil k/2\rceil + \delta}$ expected clauses, with high probability Sherali--Adams can certify that $\OBJ_{\cI}(x) \le \E[P] + \epsilon$ using $R =O_{\epsilon,\delta,k}(1)$ rounds.
\end{theorem}

This almost matches the comparable performance of sum-of-squares and spectral algorithms \cite{AOW15}, which are known to require $m \ge n^{k/2}$ clauses to certify comparable bounds in polynomial time \cite{Grigoriev01,Schoenebeck08,KMOW17}.\footnote{
The expert may notice that we require the number of clauses $m \gg n^{\lceil k/2 \rceil}$, whereas the best sum-of-squares and spectral algorithms require only $m \gg n^{k/2}$.
This is because we do not know how to relate the Sherali--Adams value of the objective function to its square (local versions of the Cauchy-Schwarz argument result in a loss).
Such a relation would allow us to apply our techniques immediately to prove that Sherali Adams matches the SOS and spectral performance for odd as well as even $k$.}
Prior to our work it was known that Sherali--Adams admits {\em weak} refutations (i.e. a certificate that $\OBJ \le 1 - o(1)$) when $m \ge n^{k/2}$, but it was conceivable (and even conjectured) that $O(1)$-rounds could not certify $\OBJ \le 1 - \delta$ for constant $\delta$ at $m = o(n^k)$.

The result above also extends to $t$-wise independent predicates as in \cite{AOW15} (see \Cref{sec:csps}).
Also, one may extract the dependence on the parameters $\epsilon,\delta$ to give nontrivial results when these parameters depend on $n$.\footnote{Though for \txor and \maxcut we have done this explicitly, for higher-arity random CSPs we have left this for the interested reader.}

\subsection{Prior work}
It is a folklore result that in random graphs with average degree $n^{\delta}$, $3$-round Sherali--Adams (SA) certifies a \maxcut value of at most $\max(1-\Omega(\delta),\frac{2}{3})$ (observed for the special case of $\delta > \frac{1}{2}$ in \cite{BarakHHS11,PT94}); this is simply because of concentration phenomena, since most edges participate in roughly the same number of odd cycles of length $O(\frac{1}{\delta}) \ge 3$, after which one can apply the triangle inequality.
However this observation does not allow one to take the refutation strength independent of the average degree.

There are some prior works examining the performance of Sherali--Adams hierarchies on random CSPs.
The work of de~la~Vega and Mathieu \cite{dlVM07} shows that in {\em dense} graphs, with average degree $\Omega(n)$, $O(1)$-rounds of Sherali--Adams certifies tight bounds on \maxcut.
This result heavily leverages the density of the graph, and is consistent with the fact that dense \maxcut is known to admit a PTAS \cite{FK96}.
However, their work does not give $O(1)$-round approximations better than $2-o(1)$ for graphs with $o(n^{2})$ edges.

Another relevant line of work is a series of LP hierarchy lower bounds (both for Sherali--Adams and for the weaker Lov\'asz-Schrijver hierarchy) for problems such as \maxcut, vertex cover, and sparsest cut, including \cite{AAT11,AroraBLT06,dlVM07,SchoenebeckTT07}, and culminating in the already mentioned result of Charikar, Makarychev and Makarychev; in \cite{CMM09}, they give subexponential lower bounds on the number of rounds of Sherali--Adams required to strongly refute \maxcut in random regular graphs.
Initially, one might expect that this result could be strengthened to prove that sparse random graphs require almost-exponential-sized LPs to refute \maxcut; our result demonstrates instead that \cite{CMM09} is almost tight.

We also mention the technique of {\em global correlation rounding} in the sum-of-squares hierarchy, which was used to give subexponential time algorithms for unique games \cite{BarakRS11} and polynomial-time approximations to max-bisection \cite{RaghavendraT12}.
One philosophical similarity between these algorithms and ours is that both relate local properties (correlation among edges) to global properties (correlation of uniformly random pairs).
But \cite{BarakRS11,RaghavendraT12} use the fact that the relaxation is an SDP (whereas our result is interesting {\em because} it is in the LP-only setting), and the ``conditioning'' steps that drive their algorithm are a fundamentally different approach.

There are many prior works concerned with certifying bounds on random CSPs, and we survey only some of them here, referring the interested reader to the discussion in \cite{AOW15}.
The sequence of works \cite{Grigoriev01,Schoenebeck08,KMOW17} establishes sum-of-squares lower bounds for refuting any random constraint satisfaction problem, and these results are tight via the SOS algorithms of \cite{AOW15,RaghavendraRS17}. 
The upshot is that for \sat{$k$} and \xor{$k$},\footnote{This is more generally true for any CSP that supports a $k$-wise independent distribution over satisfying assignments.} $\omega(1)$ rounds of SOS are necessary to strongly refute an instance with $m = o(n^{k/2})$ clauses, and $O(1)$ rounds of SOS suffice when $m = \tilde\Omega(n^{k/2})$.
Because SOS is a tighter relaxation than Sherali--Adams, the lower bounds \cite{Grigoriev01,Schoenebeck08,KMOW17} apply; our work can be seen to demonstrate that SA does not lag far behind SOS, strongly refuting with $O(1)$ rounds as soon as $m = \Omega(n^{\lceil k/2\rceil + \delta})$ for any $\delta > 0$.

In a way, our result is part of a trend in anti-separation results for SDPs and simpler methods for pseudorandom and structured instances.
For example, we have for planted clique that the SOS hierarchy performs no better than the Lov\'asz-Schrijver+ hierarchy \cite{FeigeK03,BarakHKKMP16}, and also no better than a more primitive class of estimation methods based on local statistics (see e.g. \cite{RaghavendraSS18} for a discussion).
Similar results hold for problems relating to estimating the norms of random tensors \cite{HopkinsKPRSS17}.
Further, in \cite{HopkinsKPRSS17} an equivalence is shown between SOS and spectral algorithms for a large class of average-case problems.
Our result shows that for random CSPs, the guarantees of linear programs are surprisingly not far from the guarantees of SOS.

Finally, we mention related works in extended formulations.
The sequence of works \cite{CLRS16,KothariMR17} show that SA lower bounds for CSPs imply lower bounds for any LP relaxation; the stronger (and later) statement is due to \cite{KothariMR17}, who show that subexponential-round integrality gaps for CSPs in the Sherali--Adams hierarchy imply subexponential-size lower bounds for any LP.
These works are then applied in conjunction with \cite{Grigoriev01,Schoenebeck08,CMM09} to give subexponential lower bounds against CSPs for any LP; our results give an upper limit to the mileage one can get from these lower bounds in the case of \maxcut, as we show that the specific construction of \cite{CMM09} cannot be strengthened much further.

\subsection{Techniques}
Our primary insight is that while Sherali--Adams is unable to reason about spectral properties {\em globally}, it does enforce that every set of $R$ variables behave {\em locally} according to the marginals of a valid distribution, which induces {\em local} spectral constraints on every subset of up to $R$ variables.

At first, it is unclear how one harnesses such local spectral constraints.
But now, suppose that we are in a graph with a small spectral radius.
This implies that random walks mix rapidly, in say $t$ steps, to a close-to-uniform distribution.  
Because a typical pair of vertices at distance $t$ is distributed roughly as a uniformly random pair of vertices, any subset of $R$ vertices which contains a path of length $t$ already allows us to relate global and local graph properties.

To see why this helps, we take for a moment the ``pseudoexpectation'' view, in which we think of the $R$-round Sherali--Adams as providing a proxy for the degree-$R$ moments of a distribution over \maxcut solutions $x \in \{\pm 1\}^n$, with \maxcut value
\begin{equation}\label{eq:obj-overview}
{\text\maxcut}(G) = \tfrac{1}{2} -\tfrac{1}{2} \E_{(u,v) \in E(G)} \pE[x_u x_v],
\end{equation}
where $\pE[x_ux_v]$ is the ``pseudo-correlation'' of variables $x_u,x_v$.
Because there is no globally consistent assignment, the pseudo-correlation $\pE[x_ux_v]$ for vertices $u,v$ sampled uniformly at random will be close to $0$.\footnote{This is implicit in our proof, but intuitively it should be true because e.g. $u,v$ should be connected by equally many even- and odd-length paths.}
But in any fixed subgraph of size $\Omega(t)$, enforcing $\pE[x_ux_v] \approx 0$ for pairs $u,v$ at distance $t$ has consequences, and limits the magnitude of correlation between pairs of adjacent vertices as well.
In particular, because the pseudo-second moment matrix $\pE [x_S x_S^\top]$ for $x_S$ the restriction of $x$ to a set $S$ of up to $R$ vertices must be PSD, forcing some entries to $0$ gives a constraint on the magnitude of edge correlations.

For example, suppose for a moment that we are in a graph $G$ with $t = 2$, and that $S$ is a star graph in $G$, given by one ``root'' vertex $r$ with $k \le R-1$ children $U = \{u_1,\ldots,u_k\}$, and call $X = \pE[x_S x_S^\top] \succeq 0$.  
Notice that pairs of distinct children $u_i,u_j$ are at distance $t = 2$ in $S$.
If we then require $\pE[x_{u_i}x_{u_j}] = 0$ for every $u_i \neq u_j$, the only nonzero entries of $X$ are the diagonals (which are all $\pE[ x_i^2] = 1$), and the entries corresponding to edges from the root to its children, $(r,u_i)$, which are $\pE[x_r x_{u_i}]$.
Now defining the vector $c \in \R^S$ with a $1$ at the root $r$, $c_r = 1$ and $\alpha$ on each child $u \in U$, $c_{u} = \alpha$, we have from the PSDness of $X$ that 
\[
0 
\le c^\top X c 
= \|c\|_2^2 + \sum_{u \in U} 2 c_r c_u \cdot \pE[x_r x_u]
= (1 + \alpha^2 k) + 2\alpha k \E_{(u,v) \in E(S)} \pE[x_ux_v].
\]
Choosing $\alpha = k^{-1/2}$, this implies that within $S$, the average edge correlation is lower bounded by $ \E_{(u,v) \in E[S]} \pE[x_u x_v] \ge -k^{-1/2}$.
Of course, for a given star $S$ we cannot know that $\pE[x_{u_i} x_{u_j}] = 0$, but if we take a well-chosen weighted average over {\em all} stars, this will (approximately) hold on average.

Our strategy is to take a carefully-chosen average over specific subgraphs $S$ of $G$ with $|S| = \Omega(t)$.
By our choice of distribution and subgraph, the fact that the subgraphs {\em locally} have PSD pseudocorrelation matrices has consequences for the {\em global} average pseudocorrelation across edges, which in turn gives a bound on the objective value \cref{eq:obj-overview}.
This allows us to show that Sherali--Adams certifies much better bounds than we previously thought possible, by aggregating local spectral information across many small subgraphs.

\subsection*{Organization}
We begin with technical preliminaries in \Cref{sec:prelims}.
In \Cref{sec:spiders} we prove our main result.
\Cref{sec:lbs} establishes a mild lower bound for arbitrary graphs.
Finally, \Cref{sec:csps} applies \Cref{thm:main} to the refutation of arbitrary Boolean CSPs.

\section{Setup and preliminaries}\label{sec:prelims}
We begin by recalling preliminaries and introducing definitions that we will rely upon later.

\subsection{Random walks on undirected graphs}
Here, we recall some properties of random walks in undirected graphs that will be of use to us.
\begin{definition}
    Let $G = (V,E)$ be an undirected finite graph, with parallel edges and self-loops allowed\footnote{Self-loops count as ``half an edge'', and contribute~$1$ to a vertex's degree.}, and with no isolated vertices.  The \emph{standard random walk on~$G$} is the Markov chain on $V$ in which at each step one follows a uniformly random edge out of the current vertex.  For $u \in V$, we use the notation $\bv \sim u$ to denote that $\bv$ is the result of taking one random step from~$u$.
\end{definition}
\begin{definition}
    We write $K$ for the transition operator of the standard random walk on~$G$.  That is, $K$ is obtained from the adjacency matrix of~$G$ by normalizing the $u$th row by a factor of $1/\deg(u)$.
\end{definition}
\begin{definition}
    We write $\pi$ for the probability distribution on $V$ defined by $\pi(v) = \frac{\deg(v)}{2|E|}$.  As is well known, this is an invariant distribution for the standard random walk on~$G$, and this Markov chain is reversible with respect to~$\pi$.  For $\bu \sim \pi$ and $\bv \sim \bu$, the distribution of $(\bu, \bv)$ is that of a uniformly random (directed) edge from~$E$.  We will also use the notation $\pi_* = \min_{v \in V}\{\pi(v)\}$.
\end{definition}
\begin{definition}
    For $f, g : V \to \R$ we use the notation $\ippi{f}{g}$ for $\E_{\bu \sim \pi}[f(\bu)g(\bu)]$.  This is an inner product on the vector space~$\R^V$; in case $G$ is regular and hence $\pi$ is the uniform distribution, it is the usual inner product scaled by a factor of $1/|V|$.  It holds that
    \begin{equation} \label{eqn:K1}
        \ippi{f}{Kg} = \ippi{Kf}{g} = \E_{(\bu,\bv) \sim E}[f(\bu)g(\bv)].
    \end{equation}
\end{definition}
\begin{definition}
    A \emph{stationary $d$-step walk} is defined to be a sequence $(\bu_0, \bu_1, \dots, \bu_d)$ formed by choosing an initial vertex $\bu_0 \sim \pi$, and then taking a standard random walk, with $\bu_{t} \sim \bu_{t-1}$. Generalizing \Cref{eqn:K1}, it holds in this case that
    \[
        \E[f(\bu_0)g(\bu_d)] = \ippi{f}{K^d g}.
    \]
\end{definition}

\subsection{Tree-indexed random walks}

To prove our main theorem we define a class of homomorphisms we call {\em tree-indexed random walks}.

\begin{definition}  \label{def:T-walk}
    Suppose we have a finite undirected tree with vertex set~$T$.  A \emph{stationary $T$-indexed random walk in~$G$} is a random homomorphism $\bphi : T \to V$ defined as follows: First, root the tree at an arbitrary vertex $i_0 \in T$.  Next, define $\bphi(i_0) \sim \pi$.  Then, independently for each ``child'' $j$ of~$i_0$ in the tree, define $\bphi(j) \sim \bphi(i_0)$; that is, define $\bphi(j) \in V$ to be the result of taking a random walk step from $\bphi(i_0)$.  Recursively repeat this process for all children of $i_0$'s children, etc., until each vertex $k \in T$ has been assigned a vertex $\bphi(k) \in V$.
\end{definition}
We note that the homomorphism $\bphi$ defining the $T$-indexed random walk need not be injective.
Consequently, if $T$ is a tree with maximum degree $D$, we can still have a $T$-indexed random walk in a $d$-regular graph with $d < D$.

The following fact is simple to prove; see, e.g.,~\cite{LP17}.
\begin{fact}                                        \label{fact:indep-of-root}
    The definition of $\bphi$ does not depend on the initially selected root $i_0 \in T$. Further, for any two vertices $i,j \in T$ at tree-distance~$d$, if $i = i_0, i_1, \dots, i_d = j$ is the unique path in the tree between them, then the sequence $(\bphi(i_0), \bphi(i_1), \dots, \bphi(i_d))$ is distributed as a stationary $d$-step walk in~$G$.
\end{fact}

\subsection{2XOR and signed random walks}

The \txor constraint satisfaction problem is defined by instances of linear equations in $\F_2^n$.
For us it will be convenient to associate with these instances a graph with signed edges, and on such graphs we perform a slightly modified random walk.

\begin{definition}
    We assume that for each vertex pair $(u,v)$ where $G$ has edge, there is an associated \emph{sign} $\xi_{uv} = \xi_{vu} \in \{\pm 1\}$.\footnote{If $G$ has multiple $(u,v)$ edges, we think of them as all having the same sign.}  We arrange these signs into a symmetric matrix~$\Xi = (\xi_{uv})_{uv}$.  If $G$ has no $(u,v)$ edge then the entry $\Xi_{uv}$ will not matter; we can take it to be~$0$.
\end{definition}
\begin{definition}
    We write $\ol{K} = \Xi \circ K$ for the \emph{signed transition operator}.  The operator $\ol{K}$ is self-adjoint with respect to $\ippi{\cdot}{\cdot}$, and hence has real eigenvalues.  It also holds that
    \begin{equation} \label{eqn:Kbar1}
        \ippi{f}{\ol{K}g} = \ippi{\ol{K}f}{g} = \E_{(\bu,\bv) \sim E}[\xi_{\bu\bv} f(\bu)g(\bv)].
    \end{equation}
\end{definition}
\begin{definition}
    We may think of $G$ and $\Xi$ as defining a \txor constraint satisfaction problem (CSP), in which the task is to find a labeling $f : V \to \{\pm 1\}$ so as to maximize the fraction of edges $(u,v) \in E$ for which the constraint $f(u)f(v) = \xi_{uv}$ is satisfied.  The fraction of satisfied constraints is
    \begin{equation}
        \E_{(\bu, \bv) \sim E} \left[\tfrac12 + \tfrac12 \xi_{\bu\bv} f(\bu) f(\bv)\right] = \tfrac12 + \tfrac12 \ippi{f}{\ol{K}f}.\label{eqn:KbarSigned}
    \end{equation}
    We will typically ignore the $\frac12$'s and think of the \txor CSP as maximizing the quadratic form $\ippi{f}{\ol{K}f}$.  When all signs in the matrix $\Xi$ are $-1$, we refer to this as the \maxcut CSP.
\end{definition}
\begin{definition}  \label{def:ssw}
    We say that a \emph{signed} stationary $d$-step walk is a sequence of pairs $(\bu_t, \bsigma_t) \in \{\pm 1\} \times V$ for $0 \leq t \leq d$, chosen as follows: first, we choose a stationary $d$-step walk $(\bu_0, \bu_1, \dots, \bu_d)$ in~$G$; second, we choose $\bsigma_0 \in \{\pm 1\}$ uniformly at random; finally, we define $\bsigma_{t} = \bsigma_{t-1} \xi_{\bsigma_{t-1} \bsigma_t}$.    Generalizing \Cref{eqn:Kbar1}, it holds in this case that
    \[
        \E[\bsigma_0 f(\bu_0) \bsigma_d g(\bu_d)] = \ippi{f}{\ol{K}^d g}.
    \]
\end{definition}
\begin{definition}  \label{def:sst}
    We extend the notion from \Cref{def:T-walk} to that of a \emph{signed} stationary $T$-indexed random walk in~$G$.  Together with the random homomorphism $\bphi : T \to V$, we also choose a random signing $\bsigma : T \to \{\pm 1\}$ as follows: for the root $i_0$, the sign $\bsigma(i_0) \in \{\pm 1\}$ is chosen uniformly at random; then, all other signs are deterministically chosen --- for each $j$ of $i_0$ we set $\bsigma(j) = \xi_{i_0 j} \bsigma(i_0)$, and in general $\bsigma(k) = \xi_{k'k} \bsigma(k)$ where $k'$ is the parent of~$k$.  Again, it is not hard to show that the definition of $(\bphi, \bsigma)$ does not depend on the choice of root~$i_0$, and that for any path $i_0, i_1, \dots, i_d$ of vertices in the tree, the distribution of $(\bphi(i_0), \bsigma(i_0)), (\bphi(i_1), \bsigma(i_1)), \dots (\bphi(i_d), \bsigma(i_d))$ is that of a signed  stationary $d$-step walk in~$G$.
\end{definition}

\subsection{Proof systems}\label{sec:proof-systems}
Our central object of study is the Sherali--Adams proof system, but our results also apply to a weaker proof system which we introduce below.
\begin{definition}
    We define the \emph{$R$-local, degree-$D$ (static) sum of squares (SOS) proof system} over indeterminates $\X_1, \dots, \X_n$ as follows.  The ``lines'' of the proof are real polynomial inequalities in $\X = (\X_1, \dots, \X_n)$.  The ``default axioms'' are any real inequalities of the form $p(\X_{u_1}, \dots, \X_{u_R})^2 \geq 0$, where $p$ is a polynomial in at most $R$ variables and of degree at most $D/2$.  The ``deduction rules'' allow one to derive any nonnegative linear combination of previous lines.  This is a sound proof system for inequalities about $n$ real numbers $\X_1, \dots, \X_n$.

    In addition to the default axioms, one may also sometimes include problem-specific ``equalities'' of the form $q(\X) = 0$. In this case, one is allowed additional axioms of the form $q(\X)s(\X) \gtreqless 0$ the polynomial $q(\X)s(\X)$ depends on at most $R$ indeterminates and has degree at most~$D$.
\end{definition}
\begin{fact}
    The case of $R = \infty$ (equivalently, $R = n$) corresponds to the well-known \emph{degree-$D$ SOS proof system}.
\end{fact}
\begin{definition}
    Suppose one includes the \emph{Boolean equalities}, meaning $\X_u^2 - 1 = 0$ for all $1 \leq i \leq n$.\footnote{Or alternatively, $\X_u^2 - \X_u = 0$ for all $i$.}  In this case $D = \infty$ is equivalent to $D = R$, and the corresponding proof system is the well-known \emph{$R$-round Sherali--Adams proof system}.  It is well known that every inequality $p(\X_{u_1}, \dots, \X_{u_R}) \geq 0$ that is true for $\pm 1$-valued $\X_{u_1}, \dots, \X_{u_R}$ is derivable in this system.
\end{definition}
\begin{fact}
    There is a $\poly(n^R,L)$-time algorithm based on Linear Programming  for determining whether a given polynomial inequality $p(\X) \geq 0$ of degree at most~$R$ (and rational coefficients of total bit-complexity~$L$) is derivable in the $R$-round Sherali--Adams proof system.
\end{fact}
\begin{fact}
    We will often be concerned with the $R$-local, degree-$2$ SOS proof system, where all lines are quadratic inequalities.  In this case, we could equivalently state that the default axioms are all those inequalities of the form
    \begin{equation}    \label[ineq]{eqn:sa-form}
        \mathsf{x}^\top P \mathsf{x} \geq 0,
    \end{equation}
    where $\mathsf{x} = (\X_{u_1}, \dots, \X_{u_R})$ is a length-$R$ subvector of $\X$, and $P$ is an $R \times R$ positive semidefinite (PSD) matrix.
\end{fact}
\begin{remark}
    In fact, we will often be concerned with the $R$-round, degree-$2$ Sherali--Adams proof system.  Despite the restriction to $D = 2$, we only know the $\poly(n^R, L)$-time algorithm for deciding derivability of a given quadratic polynomial $p(\X) \geq 0$ (of bit-complexity~$L$).
\end{remark}

\section{2XOR certifications from spider walks}\label{sec:spiders}
In this section, we prove our main theorem: 
given a \txor or \maxcut instance on a graph $G$ with small spectral radius, we will show that the $R$-local degree-$2$ SOS proof system gives nontrivial refutations with $R$ not too large.

Our strategy is as follows: we select a specific tree $T$ of size $\propto R$, and we consider the distribution over copies of $T$ in our graph given by the $T$-indexed stationary random walk.
We will use this distribution to define the coefficients for a degree-$2$, $R$-local proof that bounds the objective value of the CSP. 
We will do this by exploiting the uniformity of the graph guaranteed by the small spectral radius, and the fact that degree-$2$ $R$-local SOS proofs can certify positivity of quadratic forms $c^\top \X|_S\X|_S^\top c$, where $X|_S$ is the restriction of $\X$ to a set $S$ of variables with $|S| \le R$ and $c \in \R^{|S|}$.

Intuitively, in the ``pseudoexpectation'' view, the idea of our proof is as follows.
When there is no globally consistent assignment, a uniformly random pair of vertices $u,v \in V$ will have pseudocorrelation close to zero.
On the other hand, if $t$-step random walks mix to a roughly uniform distribution over vertices in the graph, then pairs of vertices at distance $t$ will also have pseudocorrelation close to zero.
But also, in our proof system the degree-$2$ pseudomoments of up to $R$ variables obey a positive-semidefiniteness constraint.
By choosing the tree $T$ with diameter at least $t$, while also choosing $T$ to propagate the effect of the low-pseudocorrelation at the diameter to give low-pseudocorrelation on signed edges, we show that the proof system can certify that the objective value is small.
Specifically, we will choose $T$ to be a {\em spider graph}:

\begin{definition}
    For integers $k, \ell \in \N^+$, we define a \emph{$(k,\ell)$-spider graph} to be the tree formed by gluing together $k$ paths of length~$\ell$ at a common endpoint called the \emph{root}.  This spider has $k\ell+1$ vertices and diameter~$2\ell$.
\end{definition}
While we were not able to formally prove that the spider is the optimal choice of tree, intuitively, we want to choose a tree that maximizes the ratio of the number of pairs at maximum distance (since such pairs relate the local properties to the global structure) to the number of vertices in the tree (because we need to take our number of rounds $R$ to be at least the size of the tree).
Among trees, the spider is the graph that maximizes this ratio.

Let us henceforth fix a $(k,\ell)$-spider graph, where the parameters $k$ and $\ell$ will be chosen later.  We write $S$ for the vertex set of this tree (and sometimes identify~$S$ with the tree itself).
\begin{definition}  \label{def:Ad}
    For $0 \leq d \leq 2\ell$, we define the matrix $A^{(d)} \in \R^{S \times S}$ to be the ``distance-$d$'' adjacency matrix of the spider; i.e., $A^{(d)}_{ij}$ is~$1$ if $\dist_S(i,j) = d$ and is~$0$ otherwise.  (We remark that $A^{(0)}$ is the identity matrix.)
\end{definition}
The following key technical theorem establishes the existence of a matrix $\psdtwo$ which will allow us to define the coefficients in our $R$-local degree-$2$ SOS proof. 
It will be proven in \Cref{sec:spider-technical}:
\begin{theorem}                                     \label{thm:spider-technical}
    For any parameter $\alpha \in R$, there is a PSD matrix $\psdtwo = \psdtwo_\alpha \in \R^{S \times S}$ with the following properties:
    \begin{align*}
        \la \psdtwo, A^{(0)} \ra &= 1 + \frac{1}{2k} \alpha^{2\ell} +  \frac{1}{k-1} \frac{\alpha^{2\ell} - \alpha^2}{\alpha^2-1}, \\
        \la \psdtwo, A^{(1)} \ra &= \alpha, \\
        \la \psdtwo, A^{(d)} \ra &= 0 \quad \text{for } 1 < d < 2\ell, \\
        \la \psdtwo, A^{(2\ell)} \ra &=  \frac{1-1/k}{2}\alpha^{2\ell}.
    \end{align*}
    Here we are using the notation $\la B,C \ra$ for the ``matrix (Frobenius) inner product'' $\Tr(B^\top C)$.
\end{theorem}
\begin{corollary}                                       \label{cor:spider-technical}
    Assuming that $k \geq 3^\ell$ and taking $\alpha = k^{1/2\ell}$, the PSD matrix $\psdtwo$ satisfies
    \[
        3/2 \leq \la \psdtwo, A^{(0)} \ra \leq 2, \quad \la \psdtwo, A^{(1)} \ra = k^{1/2\ell}, \quad \la \psdtwo, A^{(d)} \ra = 0 \text{ for $1 < d < 2\ell$}, \quad \la \psdtwo, A^{(2\ell)} \ra = \tfrac12(k-1).
    \]
\end{corollary}

We will also use the following small technical lemma:
\begin{lemma}                                       \label{lem:wack}
    Let $M \in \R^{V \times V}$ and recall $\pi_* = \min_{v \in V} \{\pi(v)\} > 0$.  Then the $2$-local, degree-$2$ SOS proof system can derive
    \[
        \E_{\bu \sim \pi} \sum_{v \in V} M_{\bu v}  \X_{\bu} \X_v \leq \pi_*^{-1/2} \|M\|_2 \E_{\bu \sim \pi} \X_{\bu}^2.
    \]
\end{lemma}
\begin{proof}
    The proof system can derive the following inequality for any $\gamma > 0$, since the difference of the two sides is a perfect square:
    \[
        M_{u v}  \X_u \X_v \leq \frac{M_{uv}^2}{2\gamma \pi(v)}  \X_u^2 +  \frac{\gamma \pi(v)}{2} \X_v^2.
    \]
    Thus it can derive
    \begin{equation}    \label[ineq]{eqn:sticky2}
        \E_{\bu \sim \pi} \sum_{v \in V} M_{\bu v}  \X_{\bu} \X_v \leq \E_{\bu \sim \pi} \X_{\bu}^2 \sum_{v \in V} \frac{M^2_{\bu v}}{2\gamma \pi(v)}   +  \frac{\gamma}{2} \E_{\bv \sim \pi} \X_v^2.
    \end{equation}
    We'll take $\gamma = \pi_*^{-1/2} \|M\|_2$.  Since we can certainly derive $a \X_u^2 \leq b \X_u^2$  whenever $a \leq b$, we see that it suffices to establish
    \[
        \sum_{v \in V} \frac{M^2_{\bu v}}{2\gamma \pi(v)} \leq \frac{\gamma}{2}
    \]
    for every outcome of~$\bu$.  But this is implied by $\sum_{v} M^2_{\bu v} \leq (\pi(v)/\pi_*)\|M\|^2_2$ for all $v \in V$, which is indeed true.
\end{proof}

We can now prove the following main theorem:
\begin{theorem}                                     \label{thm:main-2xor}
    Given parameters $k \geq 3^\ell$, let $R = k\ell + 1$ and define
    \[
        \beta = \frac{k\pi_*^{-1/2}}{2k^{1/2\ell}} \specrad(\ol{K})^{2\ell} + \frac{2}{k^{1/2\ell}}.
    \]
    Then $R$-local, degree-$2$ SOS can deduce the bound ``$\specrad(\ol{K}) \leq \beta$''; more precisely, it can deduce the two inequalities
    \[
        -\beta \ippi{\X}{\X} \leq \ippi{\X}{\ol{K}\X}  \leq \beta \ippi{\X}{\X}.
    \]
\end{theorem}
Before proving this theorem, let us simplify the parameters. 
For any $\epsilon > 0$, we can choose $\ell$ to be the smallest integer so that $(\frac{1}{\epsilon} \rho(\ol{K}))^{2\ell} \pi_*^{-1/2} \le \epsilon$, and $k = \lceil(\frac{1}{\epsilon})^{2\ell}\rceil$. 
This gives the corollary:

\begin{corollary}                                       \label{cor:main-2xor}
Suppose we have a graph $G = (V,E)$ with signed transition operator $\ol{K}$ and $\pi_* = \min_{v \in V} \frac{\deg(v)}{2|E|}$.
Given $\epsilon > \min(\pi_*^{-1/2},\rho(\ol{K}))$, take $\ell = \Big\lceil \frac{1}{4} \frac{\log (\epsilon^2 \pi_*)}{\log (\rho(\ol{K})/\epsilon)}\Big\rceil$, and take $k = \lceil (\frac{1}{\epsilon})^{2\ell}\rceil$.
Then for $R = k\ell + 1$, it holds that $R$-local degree-$2$ SOS can deduce the bound $\specrad(\ol{K}) \leq \frac{5}{2}\epsilon$.
In particular, if we think of $G, \Xi$ as a \txor CSP, it holds that $R$-round Sherali--Adams can deduce the bound $\OBJ \leq \frac12 + \frac54 \epsilon$.
\end{corollary}
\begin{proof}
Taking the parameters as above, and using that the constraints $\X_u^2 = 1$ imply that $R$-round Sherali--Adams can deduce that $\ippi{\X}{\X} = 1$ whenever $R \ge 2$, and that as noted in \cref{eqn:KbarSigned}, $\OBJ(\X) = \frac{1}{2} + \frac{1}{2}\ippi{\X}{\ol{K}\X}$, so \Cref{thm:main-2xor} gives the result.
\end{proof}

\Cref{cor:main-2xor} implies the \txor version of \Cref{thm:main} since in simple graphs, $\log \frac{1}{\pi^*} = \Theta(\log n)$.

\begin{proof}[Proof of \Cref{thm:main-2xor}]
    For our $(k,\ell)$-spider graph on $S$, let $(\bphi,\bsigma)$ be a signed stationary $S$-indexed random walk in~$G$.  Define $\bm{\ol{\mathsf{x}}}$ to be the $S$-indexed vector with $\bm{\ol{\mathsf{x}}}_i =  \bsigma(i) \X_{\bphi(i)}$.  Then letting $\psdtwo$ be the PSD matrix from \Cref{cor:spider-technical},  the $R$-local, degree-$2$ SOS proof system can derive
    \[
        \la \psdtwo , \bm{\ol{\mathsf{x}}}\bm{\ol{\mathsf{x}}}^\top\ra  = \bm{\ol{\mathsf{x}}}^\top \psdtwo \bm{\ol{\mathsf{x}}} \geq 0.
    \]
    (This is in the form of \Cref{eqn:sa-form} if we take $P = \diag(\bsigma) \psdtwo \diag(\bsigma)$.)   Furthermore, the proof system can deduce this inequality in expectation; namely,
    \begin{equation} \label[ineq]{eqn:ded1}
        \la \psdtwo, \Y \ra \geq 0, \text{ where } \Y = \E[\bm{\ol{\mathsf{x}}} \bm{\ol{\mathsf{x}}}^\top].
    \end{equation}
    Now by the discussion in \Cref{def:ssw,def:sst},
    \begin{equation} \label{eqn:Y1}
        \Y_{ij} = \E[\bsigma(i) \X_{\bphi(i)} \bsigma(j) \X_{\bphi(j)}] = \ippi{\X}{\ol{K}^{\dist_S(i,j)} \X}.
    \end{equation}
    Thus recalling the notation $A^{(d)}$ from \Cref{def:Ad},
    \begin{equation} \label{eqn:Y2}
        \Y = \sum_{d = 0}^{2\ell} \ippi{\X}{\ol{K}^{d} \X} A^{(d)},
    \end{equation}
    and hence from \Cref{eqn:ded1} we get that $R$-local, degree-$2$ SOS can deduce
    \begin{equation}    \label[ineq]{eqn:bigs}
        0 \leq \sum_{d = 0}^{2\ell} \la \psdtwo, A^{(d)}\ra \ippi{\X}{\ol{K}^{d} \X} = c_0 \ippi{\X}{\X} + k^{1/2\ell} \ippi{\X}{\ol{K}\X} + \tfrac12(k-1) \ippi{\X}{\ol{K}^{2\ell} \X},
    \end{equation}
    for some constant $3/2 \leq c_0 \leq 2$ (here we used \Cref{cor:spider-technical}). Regarding the last term, we have:
    \begin{equation}    \label{eqn:sticky}
        \ippi{\X}{\ol{K}^{2\ell} \X} = \E_{\bu \sim \pi} \sum_{v \in V} (\ol{K}^{2\ell})_{\bu v}  \X_{\bu} \X_v.
    \end{equation}
    If we cared only about the Sherali--Adams proof system with Boolean equalities, we would simply now deduce
    \[
        \E_{\bu \sim \pi} \sum_{v \in V} (\ol{K}^{2\ell})_{\bu v}  \X_{\bu} \X_v \leq \E_{\bu \sim \pi} \sum_{v \in V} \left|(\ol{K}^{2\ell})_{\bu v}\right| \leq \sqrt{|V|}\E_{\bu \sim \pi} \|\ol{K}^{2\ell}_{\bu, \cdot}\|_2 \leq \sqrt{|V|} \specrad(\ol{K}^{2\ell}) = \sqrt{|V|} \specrad(\ol{K})^{2\ell},
    \]
    and later combine this with $c_0 \ippi{\X}{\X} = c_0$.  But proceeding more generally, we instead use \Cref{lem:wack} to show that our proof system can derive
    \[
        \E_{\bu \sim \pi} \sum_{v \in V} (\ol{K}^{2\ell})_{\bu v}  \X_{\bu} \X_v \leq \pi_*^{-1/2} \specrad(\ol{K})^{2\ell} \ippi{\X}{\X}.
    \]
    Putting this into \Cref{eqn:sticky} and \Cref{eqn:bigs} we get
    \[
        \ippi{\X}{\ol{K}\X} \geq -\frac{c_0 + \tfrac12(k-1) \pi_*^{-1/2} \specrad(\ol{K})^{2\ell}}{k^{1/2\ell}} \ippi{\X}{\X} \geq -\beta \ippi{\X}{\X}.
    \]
    Repeating the derivation with $-\ol{K}$ in place of $\ol{K}$ completes the proof.
\end{proof}

\subsection{Max-Cut}
The following theorem is quite similar to \Cref{thm:main-2xor}.  In it, we allow~$K$ to have the large eigenvalue~$1$, and only certify that it has no large-magnitude negative eigenvalue.  The subsequent corollary is deduced identically to \Cref{cor:main-2xor}.
\begin{theorem}                                     \label{thm:main-max-cut}
    Given transition operator $K$ for the standard random walk on~$G$, let $K' = K - J$, where $J$ is the all-$1$'s matrix.  For parameters $k \geq 3^\ell$, let $R = k\ell + 1$ and define
    \[
        \beta = \frac{k\pi_*^{-1/2}}{2k^{1/2\ell}} \specrad(K')^{2\ell} + \frac{2}{k^{1/2\ell}}.
    \]
    (Note that $\specrad(K')$ is equal to maximum-magnitude eigenvalue of~$K$ when the trivial~$1$ eigenvalue is excluded.)  Then $2R$-local, degree-$2$ SOS can deduce the bound ``$\lambda_{\textnormal{min}}(K) \geq -\beta$''; more precisely, it can deduce the inequality
    \[
        \ippi{\X}{\ol{K}\X} \geq -\beta \ippi{\X}{\X}.
    \]
\end{theorem}
\begin{corollary}                                       \label{cor:main-max-cut}
Suppose we have a graph $G = (V,E)$ with transition operator $K$ and centered transition operator $K' = K - J$, and $\pi_* = \min_{v \in V} \frac{\deg(v)}{2|E|}$.
Given $\epsilon > \min(\pi_*^{-1/2},\rho(K'))$, take $\ell = \Big\lceil \frac{1}{4} \frac{\log (\epsilon^2 \pi_*)}{\log (\rho(\ol{K})/\epsilon)}\Big\rceil$, and take $k = \lceil (\frac{1}{\epsilon})^{2\ell}\rceil$.
Then for $R = k\ell + 1$, it holds that $2R$-local degree-$2$ SOS can deduce the bound $\specrad(K') \leq \frac{5}{2}\epsilon$.
In particular, if we think of $G$ as a \maxcut CSP, it holds that $R$-round Sherali--Adams can deduce the bound $\OBJ \leq \frac12 + \frac54 \epsilon$.
\end{corollary}
Again, \Cref{cor:main-max-cut} implies \Cref{thm:main} since in simple graphs, $\log \frac{1}{\pi^*} = \Theta(\log n)$.
\begin{proof}[Proof of \Cref{thm:main-max-cut}.]
    The proof is a modification of the proof of \Cref{thm:main-2xor}.  Letting $S$ be the $(k,\ell)$-spider vertices, instead of taking a signed stationary $S$-indexed random walk in $G$, we take two independent \emph{unsigned} stationary $S$-indexed random walks, $\bphi_1$ and $\bphi_2$.  For $j \in \{1,2\}$, define $\bm{\mathsf{x}}_{j}$ to be the $S$-indexed vector with $i$th coordinate equal to $\X_{\bphi_j(i)}$, and write $\bm{\dot{\mathsf{x}}}$ for the concatenated vector $(\bm{\mathsf{x}}_1, \bm{\mathsf{x}}_2)$.  Also, for $0 <\theta < 1$ a parameter\footnote{This parameter is introduced to fix a small annoyance; the reader might like to imagine $\theta = 1$ at first.} slightly less than~$1$, let $\psdtwo$ be the PSD matrix from \Cref{cor:spider-technical}, and define the PSD block-matrix
    \[
    \dot{\psdtwo} = \tfrac12 \left(\begin{array}{@{}c|c@{}}
                            \tfrac1\theta \psdtwo & -\psdtwo     \\[.1em]\hline
                             -\psdtwo & \theta \psdtwo
                        \end{array}\right).
    \]
    Then as before, the $2R$-local, degree-$2$ SOS proof system can derive
    \begin{equation}    \label[ineq]{eqn:yoyo}
        0 \leq \la \dot{\psdtwo}, \E \bm{\dot{\mathsf{x}}}\bm{\dot{\mathsf{x}}}^\top\ra  = \iota \la \psdtwo, \Y \ra - \la \psdtwo, \ZZ \ra, \quad \text{where } \iota = \tfrac{1/\theta + \theta}{2}, \quad \Y = \E[\bm{\mathsf{x}} \bm{\mathsf{x}}^\top], \quad \ZZ = \E[\bm{\mathsf{x}}_{1} \bm{\mathsf{x}}_{2}^\top],
    \end{equation}
    and $\bm{\mathsf{x}}$ (which will play the role of $\bm{\ol{\mathsf{x}}}$) denotes the common distribution of $\bm{\mathsf{x}}_1$ and $\bm{\mathsf{x}}_2$.  Similar to \Cref{eqn:Y1,eqn:Y2}, we now have
    \[
        \Y = \sum_{d = 0}^{2\ell} \ippi{\X}{K^{d} \X} A^{(d)},
    \]
    and by independence of $\bm{\mathsf{x}}_1$ and $\bm{\mathsf{x}}_2$ we have
    \[
        \ZZ = \ippi{1}{X}^2 \cdot J  =\ippi{1}{X}^2 \cdot \sum_{d = 0}^{2\ell} A^{(d)}.
    \]
    Thus applying \Cref{cor:spider-technical} to \Cref{eqn:yoyo}, our proof system can derive
    \begin{equation} \label[ineq]{eqn:annoy}
        0 \leq \iota \cdot\parens*{c_0 \ippi{\X}{\X} + k^{1/2\ell} \ippi{\X}{K\X} + \tfrac12(k-1) \ippi{\X}{K^{2\ell} \X}} - \E_{\bu \sim \pi} [\X_{\bu}]^2  \cdot \parens*{c_0 + k^{1/2\ell} + \tfrac12(k-1)}.
    \end{equation}
    By selecting $\theta$ appropriately, we can arrange for the factor $c_0 + k^{1/2\ell} + \tfrac12(k-1)$ on the right to equal $\iota \cdot \tfrac12(k-1)$.   Inserting this choice into \Cref{eqn:annoy} and then dividing through by~$\iota$, we conclude that the proof system can derive
    \[
        0 \leq c_0 \ippi{\X}{\X} + k^{1/2\ell} \ippi{\X}{K\X} + \tfrac{1}{2}(k-1) \parens*{ \ippi{\X}{K^{2\ell} \X} - \ippi{1}{X}^2},
    \]
    cf.~\Cref{eqn:bigs}.  Recalling now that $K$ has the constantly-$1$ function as an eigenvector, with eigenvalue~$1$, we have the identity 
    \[
        \ippi{\X}{K^{2\ell} \X} - \ippi{1}{X}^2 = \ippi{\X}{(K - J)^{2\ell} \X}.
    \]
    Now the remainder of the proof is just as in \Cref{thm:spider-technical}, with $K-J$ in place of $\ol{K}$, except we do not have the step of repeating the derivation with $-\ol{K}$ in place of $\ol{K}$.
\end{proof}

\subsection{A technical construction of coefficients on the spider} \label{sec:spider-technical}

\begin{proof}[Proof of \Cref{thm:spider-technical}]
We are considering the $(k,\ell)$-spider graph on vertex set~$S$. We write $V_t$ for the set of all vertices at distance~$t$ from the root (so $|V_0| = 1$ and $|V_t| = k$ for $1 \leq t \leq \ell$). We will be considering vectors in $\R^{S}$, with coordinates indexed by the vertex set~$S$. For $0 \leq t \leq \ell$ define the vector
\[
    \mu_t = \avg_{i \in V_t} \{\alpha^t e_i\},
\]
where $e_i = (0, \dots, 0, 1, 0, \dots, 0)$ is the vector with the $1$ in the $i$th position.  Further define vectors
\begin{align*}
    \chi &= \mu_0 + \mu_1, \\
    \phantom{\text{ for } 0 \leq t < \ell,}
        \psi_t &= \mu_t - \mu_{t+2} \text{ for } 0 \leq t < \ell,
\end{align*}
with the understanding that $\mu_{\ell+1} = 0$. Next, define the PSD matrix
\[
    \psdone = \chi \chi^\top  + \sum_{t=0}^{\ell-1} \psi_t \psi_t^\top.
\]
This will almost be our desired final matrix $\psdtwo$.  Let us now compute 
\[
    \la \psdone, A^{(d)} \ra = \chi^\top A^{(d)} \chi + \sum_{t=0}^{\ell-1} \psi_t^\top A^{(d)} \psi_t.
\]
To do this, we observe that
\[
    \mu_s^\top A^{(d)} \mu_t = \alpha^{s+t} \Pr_{\bi \sim V_s,\ \bj \sim V_t}[\dist_S(\bi,\bj) = d],
\]
and 
\begin{align*}
    \mu_0^\top A^{(d)} \mu_t =
    \mu_t^\top A^{(d)} \mu_0 &= \begin{cases}
                                                                \alpha^t & \text{if } d = t,\\
                                                                0 & \text{else;}
                                                        \end{cases}\\
    \text{and for } s, t > 0, \quad \mu_s^\top A^{(d)} \mu_{t} &= \begin{cases}
                                                                                                 (1/k)\alpha^{s+t} & \text{if } d = |s-t|,\\
                                                                                                 (1-1/k)\alpha^{s+t} & \text{if } d = s+t,\\
                                                                                                 0 & \text{else.}
                                                                                            \end{cases}
\end{align*}

\noindent From this we can compute the following (with a bit of effort):
\begin{align*}
    \la \psdone, A^{(0)} \ra &= 2 + (2/k)\alpha^2 + (2/k)\alpha^4 + \cdots + (2/k)\alpha^{2\ell-2} + (1/k) \alpha^{2\ell} \\
    \la \psdone, A^{(1)} \ra &= 2\alpha \\ 
    \la \psdone, A^{(2)} \ra &= -(2/k) \alpha^2 - (2/k)\alpha^4  -(2/k) \alpha^6 - \cdots - (2/k) \alpha^{2\ell-2}\\
    \la \psdone, A^{(2t+1)} \ra &= 0, \quad 1 \leq t < \ell \\
     \la \psdone, A^{(2t)} \ra &= 0, \quad 1 < t < \ell \\ 
    \la \psdone, A^{(2\ell)} \ra &= (1-1/k) \alpha^{2\ell}
\end{align*}
Now, for a parameter $\eta > 0$ to be chosen shortly, we finally define the PSD matrix
\[
    \psdtwo = \frac12 \psdone + \eta\mu_1\mu_1^\top.
\]
We have
\[
    \la \eta \mu_1\mu_1^\top, A^{(d)} \ra = \begin{cases}
                                                                                \eta(1/k) \alpha^2 & \text{if } d = 0, \\
                                                                                \eta(1-1/k) \alpha^2 & \text{if } d = 2.
                                                                          \end{cases}
\]
Therefore by carefully choosing
\[
    \eta = \frac{1}{k-1} \parens*{\frac{\alpha^{2\ell-2}-1}{\alpha^2-1}},
\]
we get all of the desired inner products in the theorem statement.
\end{proof}

\section{Lower Bounds}\label{sec:lbs}

In this section, we show that degree-$R$ Sherali--Adams cannot refute a random \txor or \maxcut instance better than $\frac{1}{2} + \Omega(\frac{1}{R})$.
This is a straightforward application of the framework of Charikar, Makarychev and Makarychev \cite{CMM09}.
In that work, the authors show that if every subset of $r$ points in a metric can be locally embedded into the unit sphere, then Goemans-Williamson rounding can be used to give a $\Theta(r)$-round Sherali--Adams feasible point.
The upshot is the following theorem appearing in \cite{CMM09} (where it is stated in slightly more generality, for the $0/1$ version of the cut polytope):
\begin{theorem}[Theorem 3.1 in \cite{CMM09}]\label{thm:CMM}
Let $(X, \rho)$ be a metric space, and assume that every $r = 2R + 3$ points of $(X,\rho)$ isometrically embed in the Euclidean sphere of radius $1$. 
Then the following point is feasible for $R$-rounds of the Sherali--Adams relaxation for the cut polytope:
\[
\pE[x_ix_j] =1-\frac{2}{\pi} \arccos \left(1 - \frac{1}{2}\rho(i,j)^2\right).
\]
\end{theorem}

\begin{proposition}\label{prop:lb}
In any \txor or \maxcut instance, $R$-rounds of Sherali--Adams cannot certify that 
\[
\OBJ(x) 
< \frac{1}{2} + \frac{1}{\pi R} - \frac{1}{2R^2}
\]
\end{proposition}
\begin{proof}
Suppose that we are given a \txor (equivalently, \maxcut) instance on the graph $G$, so that on each edge $(i,j) \in E(G)$ we have the constraint $x_ix_jb_{ij} = 1$ for some $b_{ij} \in \{\pm 1\}$.
Define the metric space on $(X,\rho)$ as follows: let $X = \{x_1,\ldots,x_n\}$ have a point for each vertex of $G$, and set $\rho(x_i,x_j) = \sqrt{2\left(1 - b_{ij}\frac{1}{R}\right)}$.

We claim that any $r = 2R + 3$ points of $X$ embed isometrically into the Euclidean sphere of radius $1$.
To see this, fix a set $S \subset X$, and define the $|S| \times |S|$ matrix $B_S$ so that
\[
(B_{S})_{ij} = \begin{cases}
\frac{b_{ij}}{r} & \text{if } (i,j) \in E(G), \\
0 &\text{otherwise}.
\end{cases}
\]
So long as $|S| \le r$, the matrix $M_S = \Id + B_S$ is diagonally dominant, and therefore positive semidefinite, so from the Cholesky decomposition of $M_S$ we assign to each $x_i \in S$ a vector $v_i$ so that $\|v_i\|_2 = 1$, and so that for every pair $x_i,x_j \in S$, $\|v_i - v_j\|^2 = 2 - 2b_{ij}\frac{1}{r} = \rho(i,j)^2$.

Applying \Cref{thm:CMM}, we have that the solution 
\[
\pE[x_ix_j] 
= 1-\frac{2}{\pi}\arccos\left(1 - \frac{1}{2}\cdot 2\left(1 - b_{ij}\frac{1}{r}\right)\right)
=1- \frac{2}{\pi}\arccos\left(b_{ij}\frac{1}{r}\right)
\]
is feasible.
For convenience, let $f(z) = 1-\frac{2}{\pi} \arccos(z)$.
We use the following properties of $f$:
\begin{claim}\label{claim:calc}
The function $f(z) = 1 - \frac{2}{\pi}\arccos(z)$ exhibits the rotational symmetry $f(z) = -f(-z)$, and further $f(z) \ge \frac{2}{\pi} z$ for $z \in [0,1]$.
\end{claim}
We give the proof of the claim (using straightforward calculus) below.
Now, because $f(z) = - f(-z)$, we have that
\begin{align*}
b_{ij} \cdot \pE[x_ix_j] 
=b_{ij} \cdot f\left(b_{ij}\frac{1}{r}\right)
&= f\left(\frac{1}{r}\right),\\
\intertext{
and using that for $z \in [0,1]$, $f(z) \ge \frac{2}{\pi}z \ge 0$,}
&\ge \frac{2}{\pi}\cdot \frac{1}{r}.
\end{align*}

We conclude that $R = \frac{1}{2}(r-3)$ rounds of Sherali--Adams are unable to certify that $\OBJ < \frac{1}{2} + \frac{2}{\pi}\frac{1}{2R + 3}$, as desired.
\end{proof}

\begin{proof}[Proof of \Cref{claim:calc}]
The rotational symmetry follows from simple manipulations:
\[
f(z) -(- f(-z))
= 2 - \frac{2}{\pi}(\arccos(z) + \arccos(-z))
= 2 - \frac{2}{\pi}\arccos(-1) = 0.
\]
For the second claim, we use that the derivative of $f(z) - \frac{2}{\pi}z$ is positive in the interval $[0,\frac{1}{2}]$:
\[
\frac{\partial}{\partial z} f(z) - \frac{2}{\pi} z
= \frac{2}{\pi}\frac{1}{\sqrt{1-z^2}} - \frac{1}{2} > 0 \text{ for } |z| < 1,
\]
and that at $z = 0$, $f(z) - \frac{2}{\pi}z = 0$.
\end{proof}

\section{Refutation for any Boolean CSP}\label{sec:csps}
In this section, we argue that  $R$-round Sherali--Adams can also refute any non-trivial Boolean CSP.
First, for any predicate $P:\{\pm 1\}^k\to\{0,1\}$ we define a parameterized distribution over the CSP with constraints from $P$:
\begin{definition}
Let $P:\{\pm 1\}^k \to \{0,1\}$ be a predicate.
Then we define a {\em random instance of $P$ on $n$ vertices with $m$ expected clauses} to be an instance sampled as follows:
define $p = \frac{m}{n^k}$, and for each ordered multiset $S \subset [n]$ with $|S| = k$, independently with probability $p$ we sample a uniformly random string $\zeta_S \in \{\pm 1\}^k$ and add the constraint that $P(x^{S} \odot \zeta_S) = 1$, where $\odot$ denotes the entry-wise (or Hadamard) product.
\end{definition}
This is one of several popular models, and in our case it is the most convenient to work with. 
By employing some manipulations, results from this model transfer readily to the others (see for example Appendix D of \cite{AOW15} for details).

Our result is as follows:
\begin{theorem}\label{thm:csps}
Suppose that $P:\{\pm 1\}^k \to \{0,1\}$ and that $\delta,\epsilon > 0$ are fixed constants.
Let $\E[P]$ be the probability that a random $x \in \{\pm 1\}^k$ satisfies $P$.
Then with high probability, for a random instance $\cI$ of $P$ on $n$ variables with $m \ge n^{\lceil k/2\rceil + \delta}$ expected clauses, the $R$-round Sherali--Adams proof system can certify that $\OBJ_{\cI}(x) \le \E[P] + \epsilon$ when $R =O_{\epsilon,\delta,k}(1)$ rounds.
More specifically,  $R = k \ell \left(\frac{3 \cdot 2^{k/2-1}}{\epsilon}\right)^{2\ell} + k$  for $\ell = \lceil\lceil \tfrac{k}{2}\rceil \tfrac{1}{2\delta} \rceil$.
\end{theorem}

We can also prove a more fine-grained result, to obtain strong refutation at lower clause densities when the predicate has certain properties.
\begin{definition}
We say that a predicate $P:\{\pm 1\}^k\to \{0,1\}$ is $\eta$-far from $t$-wise supporting if every $t$-wise uniform distribution has probability mass at least $\eta$ on the set of unsatisfying assignments $P^{-1}(0)$.
\end{definition}

\begin{theorem}\label{thm:csps-twise}
Suppose that $P:\{\pm 1\}^k \to \{0,1\}$ is $\eta$-far from $t$-wise supporting, and that $\delta,\epsilon > 0$.
Then with high probability, for a random instance $\cI$ of $P$ on $n$ variables and $m \ge n^{\lceil t/2\rceil + \delta}$ expected clauses, the $R$-round Sherali--Adams proof system can certify that $\OBJ_{\cI}(x) \le 1-\eta + \epsilon$ with $R =O_{\epsilon,\delta,t}(1)$ rounds.
More specifically,  $R = t \ell \left(\frac{3 \cdot 2^{t/2-1}}{\epsilon}\right)^{2\ell} + t$  for $\ell = \lceil\lceil \tfrac{t}{2}\rceil \tfrac{1}{2\delta} \rceil$.
\end{theorem}

Following the strategy introduced in \cite{AOW15}, we will do this by first refuting weighted random instances of \xor{$k$} for $k \ge 1$. 
After this, any predicate $P:\{\pm 1\}^k \to \{0,1\}$ can be decomposed according to its Fourier decomposition, which will yield a weighted sum of \xor{$t$} instances for $t \le k$, and our proof system will refute each individually.
\subsection{Higher-arity XOR}
Ultimately, we will reduce each $k$-CSP to a sum over weighted \xor{$t$} instances with $t \le k$:
\begin{definition}
Let $\cW$ be a distribution over signed integers.
We say that $\cI$ is a {\em random \xor{$k$} instance weighted according to $\cW$} if it is sampled as follows:
for each ordered multiset $S \subset [n]$ with $|S| = k$,  we take a $b_S$ to be equal to a uniformly random sample from $\cW$, and finally set the objective function to be $\sum_{S} b_S \cdot x^S$.
\end{definition}

Following the standard strategy introduced by \cite{GoerdtK01,FriedmanGK05} and subsequently honed in many works, we will reduce refuting these \xor{$t$} instances to refuting \txor instances.

\subsubsection{Even $k$-XOR}
In this case, we perform a standard transformation to view the \xor{$k$} instance as a \txor instance on super-vertices given by subsets of vertices of size $k/2$.

\begin{definition}\label{def:flattening}
Suppose $k > 1$ is an integer and $\cI$ is a \xor{$2k$} instance on $n$ variables $x_1,\ldots,x_n$, with objective $\sum_{U \in [n]^{2k}} b_U \cdot x^{U}$ where the sum is over ordered multisets $U \subset [n], |U| = 2k$. 
Then we let its {\em flattening}, $\cI_{flat}$, be the \txor instance on $n^k$ variables given by associating a new variable $y_S$ for each ordered multiset $S\subset [n], |S| = k$, and for each $U \subset [n]$ with $|U| = 2k$, choosing the partition of $U$ into the ordered multisets $S,T$ with $S$ containing the first $k$ elements and $T$ containing the last $k$, taking the objective function $\sum_{S,T} b_{U}\cdot y_S y_T$.
\end{definition}

\begin{lemma}\label{lem:flattening}
Suppose that $\cI$ is a \xor{$2k$} instance, and let $\cI_{flat}$ be the \txor instance given by its flattening.
Then if the $R$-round Sherali--Adams proof system can certify that $\OBJ_{\cI_{flat}}(x) \le c$, then the $k\cdot R$-round Sherali--Adams proof system can certify that $\OBJ_{\cI}(x) \le c$.
\end{lemma}
\begin{proof}
Every degree-$R$ Sherali--Adams proof for $\cI_{flat}$ can be transformed into a Sherali--Adams proof of degree at most $kR$ for $\cI$ by applying the transformation $y_S = \prod_{i \in S} x_i = x^S$.
Further, this transformation exactly relates the objective functions of $\cI_{flat}$ and $\cI$.
This proves the claim.
\end{proof}

If the \xor{$2k$} instances that we start with are random weighted instances, then their flattenings are also random weighted \txor instances.
\begin{claim}\label{claim:preserve}
Suppose that $\cI$ is a random \xor{$2k$} instance on $n$ vertices weighted according to $\cW$.
Then the flattening $\cI_{flat}$ is a random \txor instance on $n^{k}$ vertices weighted according to $\cW$.
\end{claim}
\begin{proof}
This fact is immediate, since the ordered multisets $U \subset [n]$, $|U| = 2k$ are in bijection with ordered pairs of multisets $S,T \subset [n]$, $|S| = |T| = k$.
\end{proof}

We will require the following proposition, which applies our main theorem in the context of random \xor{$k$} instances with random weights from well-behaved distributions.
\begin{proposition}\label{prop:small-refute}
Suppose that $\cW$ is a distribution over integers which is symmetric about the origin, and let $n,k \ge 1$ be positive integers. 
Let $\E$ denote the expectation under the measure $\cW$, and let $\sigma^2 \ge E w^2$ be a bound on the variance.
Furthermore, suppose that
\begin{itemize}
\item The expected absolute value is at least $\E|w| \gg \sigma \sqrt{\frac{\log n}{n^k}}$,
\item With high probability over $n^{2k}$ i.i.d. samples $w_1,\ldots,w_{n^{2k}} \sim \cW$, $\max_{i \in [n^{2k}]} |w_i| \le M \ll \sigma^2 n^{k}$.
\end{itemize}
Now, define
\[
\rho =O\left (\frac{\sigma \log N}{\E|w| \sqrt{n^k}} \cdot \max(1,\tfrac{M}{\sqrt{n^k}})\right).
\]
Then if $\cI$ is a random \xor{$2k$} instance on $n$ variables weighted according to $\cW$, with high probability $\cI$ has $\E|w| \cdot n^{2k} \pm O(\sigma n^{k} \sqrt{\log n})$ constraints.
Further, choosing $\ell \in \N_+$ large enough so that $n^{k/4\ell}\rho \le \frac{1}{2}\epsilon^{2\ell}$ and setting $R = 2k\cdot \ell \cdot \left(\frac{1}{\epsilon}\right)^{2\ell}$, $R$ rounds of Sherali--Adams can deduce the bound $\OBJ_\cI(x) \le \frac{1}{2} + \frac{3}{2}\epsilon$.
\end{proposition}

To prove the above, we require the following standard matrix Bernstein inequality:
\begin{theorem}[Theorem 6.6.1 in \cite{tropp2015}]\label{thm:spectral}
Let $A_1,\ldots,A_m \in \R^{N \times N}$ be independent random matrices, with $\E A_i = 0$ for all $i \in [m]$ and $\|A_i\| \le M$ for all $i \in [m]$.
Let $A = \sum_{i \in [m]} A_i$ denote their sum, and suppose that $\|\E AA^\top\|, \|E A^\top A\| \le \sigma^2$.
Then 
\[
\Pr\left(\|A\| \ge t\right) \le N \cdot \exp\left(\frac{1}{2}\frac{-t^2}{\sigma^2 + \tfrac{1}{3}Mt}\right).
\]
\end{theorem}

\begin{proof}[Proof of \Cref{prop:small-refute}]
Given a weighted \xor{$2k$} instance on $n$ variables with weights from $\cW$, we consider its flattening $\cI_{flat}$ with objective function $\OBJ(x) = \frac{1}{m}\sum_{i,j \in [N]} \frac{1}{2}(1+b_{ij} x_i x_j)$ for $m$ the absolute sum of weights, we construct its signed adjacency matrix as follows: 
first take the matrix $W$ defined so that $W_{i,j} = b_{ij}$, and obtain a new matrix $B = \frac{1}{2}(W + W^\top)$. 
For any $x$, applying \Cref{lem:flattening} we have that $\frac{1}{2}+ \frac{1}{2m} x^\top B x = \OBJ_{\cI}(x)$.

Since $\cW$ is a distribution over integers, $2B$ has signed integer entries.
We think of $2B$ as defining a multigraph $G$ on $n^k$ vertices with signed edges, so that there are $2\cdot |B_{ij}|$ multiedges between $i,j \in [n^k]$, each with sign $\sgn(B_{ij})$.
Let $2\cdot D$ be the degree matrix of $G$, let $A = |2B|$ be the adjacency matrix of $G$, let $\Xi = \sgn(B)$ be the matrix of signs of $B$, and let $K = (2D^{-1}) A = D^{-1}B \otimes \Xi$ be the transition matrix for the random walk on $G$.

To apply \Cref{cor:main-2xor}, we must upper bound the spectral radius of $\Xi \circ K = D^{-1} B$, as well as bound the minimum degree of $G$ and the total number of edges.
We will use the bound
\[
\|D^{-1}B\|_{op} \le \|D^{-1}\|_{op}\cdot \|B\|_{op} \le \frac{1}{\pi^*}\|B\|_{op}.
\]
First, we bound $\|B\|_{op}$.
Take $B'$ to be the truncated version of $B$, so that $B_{i,j}' = \sgn(B_{i,j})\cdot \max(|B_{i,j}|, M)$.
Thinking of the matrix $B'$ as the sum of $\binom{n^k}{2} + n^k$ symmetric matrices, one for each pair $i,j \in [n^k]$, satisfies the requirements of \Cref{thm:spectral}.
We have that $\E [B' B'^\top] \preceq n^k \sigma^2 \cdot \Id$, so applying \Cref{thm:spectral} with $t = O\max(\sqrt{\sigma^2 n^k\log n}, M \log n)$ we get that with high probability,
\[
\|B\|_{op} \le O\left(\max\left(\sqrt{\sigma^2 n^k\log n}, M\log n\right)\right),
\] 
where we have also used that with high probability $B = B'$ by the properties of $\cW$.

Now, we bound the sum of degrees $2m$ and the minimum degree $d_{\min}$.
We have that the total sum of the degrees is given by 
$2m = \sum_{i,j \in [n^k]} |b_{ij}|$ with $b_{ij} \sim \cW$.
By a Bernstein inequality,
\[
\Pr\left( \left|2m - n^{2k}\E|w| \right| \ge s \right)
\le 2\exp\left(-\frac{1}{2}\frac{s^2}{n^{2k} \cdot \E w^2 + \tfrac{1}{3}M s}\right),
\]
so since by assumption $\sigma^2 n^{2k} \gg M$, setting $s = O(\sigma n^k\sqrt{\log n})$ we have that with high probability 
\begin{equation}\label{eq:coef-total}
2m = n^{2k} \E |w| \pm O(\sigma n^k\sqrt{\log n}).
\end{equation}

By our assumptions on $\cW$ we have that for every $i \in [n^k]$, $\E \deg_G(i) = n^k \E|w|$.
Applying a Bernstein inequality gives that 
\[
\Pr[\deg_{G}(i) \le n^k\E|w| - t] \le \exp\left(-\frac{1}{2} \cdot \frac{t^2}{n^k\sigma^2  + \frac{1}{3}Mt}\right),
\]
so using that $M \le n^k\sigma^2$ and taking $t = O\left(\sqrt{n^k\sigma^2\log n}\right)$ we get that $d_{\min} = n^k \cdot \E|w| \pm O(\sqrt{n^k \sigma^2 \log n})$ with high probability.
This gives that with high probability,
\begin{align*}
\pi^* 
&= \frac{d_{\min}}{2m} 
\ge \frac{n^k \cdot \E |w| - O(\sigma\sqrt{n^k\log n})}{n^{2k} \cdot \E|w| + O(\sigma n^k \sqrt{\log n})}
\ge \frac{1}{n^k}\cdot (1-o(1))
\\
\|\Xi \circ K\|_{op} 
&\le \frac{\|B\|_{op}}{d_{\min}}
\le  \frac{O(\sigma \sqrt{n^k}\log n) \cdot \max(1,\tfrac{M}{\sqrt{n^k}})}{n^k\cdot \E|w| - O(\sigma \sqrt{n^k\log n})}
\le O\left(\frac{\sigma \log n}{\E|w| \sqrt{n^k}}\right) \cdot \max(1,\tfrac{M}{\sqrt{n^k}})
\end{align*}
where we have used that $\sigma \sqrt{\log n} \ll \sqrt{n^k} \E|w|$.

Now, the result follows by applying \Cref{cor:main-2xor} and \Cref{lem:flattening}.
\end{proof}

\subsubsection{Odd $k$-XOR}
For odd integers $k$, \xor{$k$} instances do not have the same natural, symmetric flattenings. 
Instead, we define what we call a {\em lift}:
\begin{definition}\label{def:lift}
Suppose $k \ge 1$ is an integer and $\cI$ is a \xor{$(2k+1)$} instance on $n$ variables $x_1,\ldots,x_n$, with objective $\sum_{U \in [n]^{2k+1}} b_U \cdot x^{U}$.
Then we let its {\em lift}, $\cI_{lift}$, be the bipartite \txor instance on parts each containing $n^{k+1}$ variables created as follows:
\begin{itemize}
\item Create new variables $w_1,\ldots,w_n$
\item For each $U \in [n]^{2k+1}$, choose a random index $i_U \in [n]$ and add modify the objective to $\sum_{U \in [n]^{2k+1}} b_U \cdot x^{U} \cdot w_{i_U}$
\item For each ordered multiset $S$ associate a new variable $y_S$, and for each ordered multiset $T \in [n]^k$ and index $i \in [n]$ associate a new variable $z_{T,i}$
We understand $y_S = \prod_{i \in S} x_i$, and $z_{T,i} = \left(\prod_{j \in T} x_j \right)\cdot w_{i}$.
\item For each $U \in [n]^{2k+1}$, we take the ordered multiset $V = (U,i_U)$ and assign it a new coefficient $b'_V = b_U$. 
For the remaining $b'_V$, we set $b'_V = 0$. 
\end{itemize}
Finally, $\cI_{lift}$ is the instance with the objective function $\sum_{S \in [n]^{k+1}, T \in [n]^k, i \in [n]} b'_{S \cup T \cup i} \cdot y_S z_{T,i} $.
\end{definition}

We obtain a statement analogous to \Cref{lem:flattening} for odd \xor{$k$}:

\begin{lemma}\label{lem:lift}
Suppose that $\cI$ is a weighted \xor{$(2k+1)$} instance, and let $\cI_{flat}$ be the bipartite \txor instance given by its flattening.
Then if the $R$-round Sherali--Adams proof system can certify that $\OBJ_{\cI_{flat}}(x) \le c$, then the $(k+1)\cdot R$-round Sherali--Adams proof system can certify that $\OBJ_{\cI}(x) \le c$.
\end{lemma}
\begin{proof}
The only modification to the proof of \Cref{lem:lift} is that for $z_{T,i}$ we substitute $z_{T,i} = x^T$ (where we have implicitly substituted $w_{i} = 1$ for all $i \in [n]$).
\end{proof}

However, the lifting procedure does not preserve the weighting distribution $\cW$, because of the step in which a random index $i_U$ is chosen to lift $U$.
For this reason, we prove an analog of \Cref{prop:small-refute}:
\begin{proposition}\label{prop:small-refute-odd}
Suppose that $\cW$ is a distribution over integers which is symmetric about the origin, and let $n,k \ge 1$ be integers. 
Let $\E$ denote the expectation under the measure $\cW$, and let $\sigma^2 \ge E w^2$ be a bound on the variance.
Furthermore, suppose that
\begin{itemize}
\item The expected absolute value is at least $\E|w| \gg \sigma \sqrt{\frac{\log n}{n^{k}}}$,
\item With high probability over $n^{2k+1}$ i.i.d. samples $w_1,\ldots,w_{n^{2k+1}} \sim \cW$, $\max_{i \in [n^{2k+1}]} |w_i| \le M \ll \sigma^2 n^k$.
\end{itemize}
Now, define
\[
\rho =O\left (\frac{\sigma \log n}{\E|w| \sqrt{n^k}} \cdot \max(1,\tfrac{M}{\sqrt{n^k}})\right).
\]
Then if $\cI$ is a random \xor{$(2k+1)$} instance on $n$ variables weighted according to $\cW$, with high probability it has $\E|w| \cdot n^{2k+1} \pm O(\sigma n^k \sqrt{\log n})$ constraints.
Furthermore, choosing $\ell \in \N_+$ large enough so that $n^{(k+1)/4\ell}\rho \le \frac{1}{2}\epsilon^{2\ell}$ and setting $R = (2k+2)\ell \cdot \left(\frac{1}{\epsilon}\right)^{2\ell}$, $R$ rounds of Sherali--Adams can deduce the bound $\OBJ_\cI(x) \le \frac{1}{2} + \frac{3}{2}\epsilon$.
\end{proposition}

\begin{proof}
The thread of the proof is the same as that of \Cref{prop:small-refute}.
We will refute $\cI_{lift}$, since by \Cref{lem:lift} this is sufficient.
We begin by associating with $\cI_{lift}$ a multigraph $G$ (which we may do because $\cW$ is a distribution over integers).
The multigraph $G$ is a bipartite graph, with one bipartition corresponding to variables $y_S$ for $S \in [n]^{k+1}$, and one bipartition corresponding to variables $z_{T,i}$ for $T \in [n]^k$ and $i \in [n]$.
We let the block matrix $B$ be (the $\tfrac{1}{2}$-scaled) signed adjacency matrix of $G$, let $\Xi$ be the matrix of signs, $D$ be the diagonal degree matrix, and $K\circ \Xi = D^{-1} B$ be the signed transition matrix of the random walk on $G$.
In order to apply \Cref{cor:main-2xor} we must bound $\|K \circ \Xi\|$ and $\pi^* = d_{\min}(G)/2m$.

First, we bound the vertex degrees.
For a vertex of the form $(T,i)$, the expected value of the incident edge $(S,T\cup i)$ is $b_{S,T} \cdot \frac{1}{n}$, where $b_{S,T}$.
The degree of $T\cup i$ is simply the sum
\[
\deg_G(T\cup i) = \sum_{S \in [n]^{k+1}} |b_{S,T\cup i}|,
\]
a sum of independent random variables with expectation $\frac{1}{n} \E|w|$ and variances $\frac{1}{n}\sigma^2$.
Applying a Bernstein inequality, we have that
\[
\Pr\left(\left|\deg_G(T\cup i) - \frac{1}{n}\cdot n^{k+1}\E[|w|]\right| \ge t\right)
\le 2\exp\left(\frac{1}{2}\frac{-t^2}{ n^k \sigma^2 + \frac{1}{3}Mt}\right),
\]
so taking $t = O(\sqrt{\sigma^2 n^k \log n})$ (and using that $M \ll n^k \sigma^2$), we have that the degree of $T\cup i$ vertices is $\deg_G(T\cup i) = n^k \E|w| \pm O(\sqrt{\sigma^2 n^k \log n})$ with high probability.

A similar argument applies to the $S$ vertices; the total degree of such a vertex is
\[
\deg_G(S) = \sum_{T \in [n]^{k}} |\sum_{i \in [n]} b_{S, T\cup i}|,
\]
since only one of the $b_{S,T \cup i}$ will be nonzero.
The inner sums are independent random variables with  mean $\E|w|$ and variance $\sigma^2$, therefore
\[
\Pr\left(\left|\deg_G(S) - n^{k}\E|w|\right| \ge t\right)
\le 2\exp\left(\frac{1}{2}\frac{-t^2}{ n^k \sigma^2 + \frac{1}{3}Mt}\right),
\]
so taking $t = O(\sqrt{\sigma^2 n^k \log n})$ (and using that $M \ll n^k \sigma^2$), we have that the degree of $S$ vertices is also $\deg_G(S) = n^k \E|w| \pm O(\sqrt{\sigma^2 n^k \log n})$ with high probability.

We finally bound $\|K \circ \Xi\| \le \|D^{-1}\|\cdot \|B\|$.
As before, $B$ is a sum of independent symmetric matrices, one for each coefficient $b_U$ from $\cI$.
That is, we can define matrices $B_{U}$ for each $U \in [n]^{2k+1}$ with $U = S,T$ for $S \in [n]^{k+1}, T \in [n]^k$ where $B_U$ has a number $b_U \sim \cW$ in one off-diagonal block entry $(S,T\cup i)$ and the other off-diagonal block entry $(T\cup i,S)$ for a randomly chosen $i \in [n]$.
Thus, $\E B_U B_U^\top$ is a diagonal matrix with $\frac{1}{n} \sigma^2$ on each diagonal of the form $(T\cup i, T\cup i)$ and $\sigma^2$ on each block diagonal of the form $(S,S)$.
We then have that $\E BB^\top \preceq n^{k} \sigma^2 \Id$, since for each $S$ there is a sum over $n^k$ matrices $B_U$ and for each $T \cup i$ there is a sum over $n^{k+1}$ matrices $B_U$.
Applying \Cref{thm:spectral} by using the same truncation trick again, we have that
\[
\|B\| \le O\left(\max\left(\sqrt{\sigma^2 n^k \log n}, M\log n\right)\right),
\]
and from this we have that with high probability,
\begin{align}
\|K \circ \Xi\| 
&\le \frac{1}{\deg_{\min}(G)} \le O\left(\frac{\sigma \log n}{\E|w|\sqrt{n^k}}\right)\cdot \max(1,\tfrac{M}{\sqrt{n^k}}),\\
\pi^* &= \frac{\deg_{\min}(G)}{2m} = \frac{1}{n^{k+1}}\cdot(1 \pm o(1))
\end{align}
After which we can apply \Cref{cor:main-2xor}.
\end{proof}

\subsection{From Boolean CSPs to $k$-XOR}
Following \cite{AOW15}, we prove \Cref{thm:csps} via reduction to {\sc XOR}. 

\begin{proof}[Proof of \Cref{thm:csps}]
Given a random instance of the CSP defined by the predicate $P$, and $p = n^{-\lfloor k/2\rfloor + \delta}$, a Bernstein inequality gives us that the number of constraints $m$ is with high probability given by $m = n^{\lceil k/2\rceil +\delta} \pm 10\sqrt{n^{\lceil k/2\rceil + \delta}\log n}$.
Set $\ell = \lceil \frac{1}{2\delta}\rceil$.

Since $P$ is a Boolean predicate, we can write $P$ in its Fourier expansion:
\begin{align*}
P(z)
&= \sum_{\alpha \subseteq [k]} \widehat P(\alpha) \prod_{i \in \alpha} z_i.
\end{align*}
Using this expansion, we re-write the objective function.
Recall that $[n]^k$ is the set of all ordered multisets of $k$ elements of $[n]$.
For each $S \in [n]^k$, let $b_S$ be the $0/1$ indicator that there is a constraint on $S$. 
Then, if the total number of constraints is $m$,
\begin{align}
\OBJ_{\cI}(x)
&= \frac{1}{m}\sum_{S \in [n]^k} b_S \cdot P(x^{S} \odot \zeta_S)\nonumber\\
&= \frac{1}{m}\sum_{S = \{i_1,\ldots,i_k\} \in [n]^k} \sum_{\alpha \subseteq [k]} b_S \cdot \widehat P(\alpha) \prod_{a \in \alpha} x_{i_a} (\zeta_S)_{i_a}\nonumber\\
&= \frac{1}{m}\sum_{\alpha \subseteq [k]} \widehat P(\alpha)\cdot \sum_{T \in [n]^{|\alpha|}} \left(\sum_{S \in [n]^k, S|_{\alpha} = T} b_S \cdot \prod_{a \in \alpha} (\zeta_S)_{i_a}\right) \cdot x^T. \label{eq:obj}
\end{align}
Now, define for each $\alpha \subseteq [k]$ with $|\alpha| = t > 0$ the \xor{$t$} instance
\begin{align*}
\cI^{\alpha}(x) 
&=\frac{1}{m}\sum_{T \in [n]^{t}} \left(\sum_{S \in [n]^k, S|_{\alpha} = T} b_S \cdot \prod_{a \in \alpha} (\zeta_S)_{i_a}\right) \cdot x^T
=\frac{1}{m}\sum_{T \in [n]^{t}} w_T \cdot x^T,
\end{align*}
where we have taken $w_T = \sum_{S \in [n]^k, S|_{\alpha} = T} b_S \cdot \prod_{a \in \alpha} (\zeta_S)_{i_a}$.
So that from \Cref{eq:obj},
\begin{equation}
\OBJ_{\cI}(x) = \sum_{\alpha \subseteq [k]} \widehat P(\alpha) \cdot \cI^{\alpha}(x).\label{eq:clean-form}
\end{equation}

Let $\cW_{n^{k-t}}$ be the distribution defined so that $w \sim \cW_t$ is a sum of $n^{k-t}$ independent variables taking value $\{\pm 1\}$ with probability $p$ and value $0$ otherwise.
Since for each $S \supseteq T$, the quantity $\prod_{a \in \alpha} (\zeta_S)_{i_a}$ is an independent uniform sign in $\{\pm 1\}$ and $b_S$ is an independent Bernoulli-$p$ variable, we have that the coefficients $w_T$ in $\cI^{\alpha}$ are i.i.d. from $\cW_{n^{k-t}}$.
The following lemma establishes some properties of $\cW_{N}$ (we will prove the lemma in \Cref{app:Ws}):
\begin{lemma}\label{lem:w-props}
Let $\cW_{N}(p)$ be the distribution defined so that $X \sim \cW_{N}$ is given by $X = \sum_{t = 1}^N Y_t \cdot Z_t$, where the $\{Y_t\}_t, \{Z_t\}_t$ are i.i.d with $Y_t \sim \Ber(p)$ and $Z_t \sim \{\pm 1\}$.
Then for $X \sim \cW_N(p)$, $\E X = 0$ and $\E X^2 = pN$.
Further, so long as $pN \ge 1$, $\E|X| \ge \frac{2}{e^{3/2}}\sqrt{pN}$, and $\Pr(|X| > 2t\sqrt{pN}) \le 2\exp\left(-t^2\right)$.
Otherwise, if $pN \le 1$, 
$\E|X|\ge \frac{1}{2e}\log\frac{1}{1-pN}$,
and
$\Pr(|X| \ge 1 + t) \le \exp\left( -\frac{1}{2} t\right)$.
\end{lemma}

From \Cref{lem:w-props}, we have that $\E w_T^2 = pn^{k-t}$, and by Cauchy-Schwarz $\E |w_T| \le \sqrt{\E w_T^2}$.
Let $m_{\alpha}$ be the total absolute weight of constraints in $\cI^{\alpha}$,
\[
m_{\alpha} = \sum_T |w_T|.
\]
Notice that in all cases, $m_{\alpha} \le m$.

Now, we show that SA can certify upper bounds on $|\cI^{\alpha}(x)|$ for every $\alpha$.
First, consider $\alpha$ with $|\alpha|= t=1$.
In this case,  Sherali--Adams with $R = 1$ can certify that
\[
\cI^{\alpha}(x)
= \frac{1}{m}\sum_{i \in [n]} w_i \cdot x_i 
\le\frac{1}{m} \sum_{i \in [n]} |w_i|
= \frac{m_{\alpha}}{m},
\]
From an application of Bernstein's inequality (the same as in the proof of \Cref{prop:small-refute}), $m_{\alpha} \le n \cdot \sqrt{\E w_T^2} + \sqrt{pn^k \log n}$ with high probability whenever $pn^k \gg \sqrt{pn^{k-1}}$, and applying our bound on $m$ we conclude that with high probability SA will certify that 
\[
\cI^{\alpha}(x) \le \frac{n\cdot \sqrt{pn^{k-1}}(1+o(1))}{pn^k(1\pm o(1))} \le \frac{2}{\sqrt{pn^{k-1}}}
\]
The lower bound on $\cI^{\alpha}(x)$ follows from identical reasoning with its negation, so we can conclude that with high probability SA can certify that
\[
|\cI^{\alpha}(x)| \le \frac{2}{\sqrt{pn^{k-1}}} = o(1).
\]

Now, we tackle $\alpha$ with $|\alpha| = t$ for $2 \le t\le k$.
We will verify that the conditions of \Cref{prop:small-refute,prop:small-refute-odd} hold.
First, consider the $\alpha$ with $|\alpha| = t$ for $pn^{k-t} \ge 1$.
From \Cref{lem:w-props}, in this case we have that $\E|w_T| \ge \frac{2}{e^{3/2}} \sqrt{pn^{k-t}}$, and with high probability, $|w_T| \le O(\sqrt{tpn^{k-t}\log n})$ for all $T \in [n]^t$.
Letting $M = O(\sqrt{tpn^{k-t}\log n})$, and $\sigma^2 = pn^{k-t}$, we meet the conditions for \Cref{prop:small-refute}:
\begin{align*}
M &\le O(\sqrt{pn^{k-t}}) \ll pn^{k-t} \cdot n^{\lfloor t/2\rfloor} = \sigma^2n^{\lfloor t/2\rfloor}\\
\E|w_T| &\ge \frac{2}{e^{3/2}} \sqrt{pn^{k-t}} \gg \sqrt{\frac{pn^{k-t}\log n}{n^{\lfloor t/2\rfloor}}}=\sqrt{\frac{\sigma^2\log n }{n^{\lfloor t/2\rfloor}}} 
\end{align*}
so long as $pn^{k-t} \ge 1$, which we have assumed.
So applying \Cref{prop:small-refute,prop:small-refute-odd} to both $\frac{m}{m_{\alpha}}\cI^{\alpha}$ and $-\frac{m}{m_{\alpha}}\cI^{\alpha}$, we have
\[
\rho = O\left(\frac{\sigma\log n}{\E|w_T| \sqrt{n^{\lfloor t/2\rfloor}}}\right) \cdot \max(1,\tfrac{M}{\sqrt{n^{\lfloor t/2\rfloor}}}) 
\le O\left(\frac{\log n}{\sqrt{n^{\lfloor t/2\rfloor}}}\right) \cdot \max\left(1,\frac{\sqrt{pn^{k-t}}}{n^{\lfloor t/2\rfloor}}\right)
\]
and so long as $\frac{m_{\alpha}}{m}\cdot n^{\lceil t/2 \rceil / 4\ell} \rho \le \frac{1}{2} \epsilon^{2\ell}$, with high probability over $\cI^{\alpha}$, $t(\ell r + 1)$ rounds of Sherali--Adams certify that $|\cI^{\alpha}(x)| \le \frac{3}{2}\epsilon$.
We confirm that 
\[
\frac{m_{\alpha}}{m} \cdot n^{\lceil t/2 \rceil/4\ell} \cdot \rho
= \frac{\sqrt{pn^{k-t}} \cdot n^t}{pn^k} \cdot n^{\lceil t/2 \rceil/4\ell} \cdot \rho
\ll o(1),
\]
whenever $t \ge 1$ and $\ell \ge 1$.

Finally, we handle $\alpha$ with $|\alpha| = t$ for $2 \le t$ and $pn^{k-t} < 1$.
From \Cref{lem:w-props} we have $\E|w_T| \ge \frac{1}{e} \log \frac{1}{1-pn^{k-t}}$, and with high probability, $|w_T| \le 4\log n$ for all $T \in [n]^t$.
Taking $M = 4\log n$, we have that we meet the conditions of \Cref{prop:small-refute,prop:small-refute-odd}
\begin{align*}
M &\le 4\log n \ll n^{\lceil k/2 \rceil - \lceil t/2\rceil + \delta} = pn^{k-t}n^{\lfloor t/2 \rfloor} = \sigma^2 n^{\lfloor t/2\rfloor},
\\
\E|w_T| &\ge \frac{1}{e}\log\frac{1}{1-pn^{k-t}} \ge \frac{1}{e} pn^{k-t} \gg \sqrt{\frac{pn^{k-t}\log n}{n^{\lfloor t/2\rfloor}}}=\sqrt{\frac{\sigma^2\log n}{n^{\lfloor t/2 \rfloor}}},
\end{align*}
where the last inequality is true whenever
$p n^{k-t + \lfloor t/2\rfloor} = n^{\lceil k/2\rceil - \lceil t/2\rceil + \delta} \gg \log n$,
which we have by assumption.
So applying \Cref{prop:small-refute,prop:small-refute-odd} to both $\frac{m}{m_{\alpha}}\cI^{\alpha}$ and $-\frac{m}{m_{\alpha}}\cI^{\alpha}$, we have that for 
\[
\rho = O\left(\frac{\sqrt{pn^{k-t}}\log n}{pn^{k-t} \sqrt{n^{\lfloor t/2\rfloor}}}\right) 
= O(\log n)\cdot \sqrt{\frac{1}{pn^{k- \lceil t/2\rceil}}}
= O \left(\frac{\log n}{\sqrt{n^{\lceil k/2 \rceil - \lceil t/2\rceil + \delta}}}\right),
\]
so long as $\frac{m_{\alpha}}{m}n^{\lceil t/2 \rceil /4\ell} \rho \le \frac{1}{2}\epsilon^{2\ell}$, $R = t(\ell r + 1)$ rounds of Sherali--Adams certify that $|\cI^{\alpha}(x)| \le \frac{3}{2}\epsilon$ with high probability.
Verifying,
\[
\frac{m_{\alpha}}{m}\cdot n^{\lceil t/2 \rceil/4\ell} \cdot O\left(\frac{\log n}{\sqrt{n^{\lceil k/2 \rceil - \lceil t/2 \rceil + \delta}}}\right)
= \frac{1}{\sqrt{pn^{k-t}}}\cdot n^{\lceil t/2 \rceil/4\ell} \cdot O\left(\frac{\log n}{\sqrt{n^{\lceil k/2 \rceil - \lceil t/2 \rceil + \delta}}}\right),
\]
which is maximized at $t = k$. 
By our choice of $\ell$, the condition holds.

We therefore have (using Parseval's identity and $\|x\|_1 \le \sqrt{k}\|x\|_2$ for $x \in \R^k$ to simplify \Cref{eq:clean-form}) that the same number of rounds certifies that
\[
\OBJ_{\cI}(x) \le \sum_{\alpha \subset [k]} \widehat P(\alpha) \cdot \cI^{\alpha}(x)
\le  \widehat P(\emptyset) + \sqrt{2^k} \frac{3}{2}\epsilon,
\]
as desired.
\end{proof}

Using arguments analogous to the above along with the reasoning outlined in Theorem 4.9, proof 2 and Claim 6.7 from \cite{AOW15}, we can also prove \Cref{thm:csps-twise}.

\section*{Acknowledgments}
We thank Luca Trevisan for helpful comments, and Boaz Barak for suggesting the title.
We also thank the Schloss Dagstuhl Leibniz Center for Informatics (and more specifically the organizers of the CSP Complexity and Approximability workshop), as well as the Casa Mathem\'{a}tica Oaxaca (and more specifically the organizers of the Analytic Techniques in TCS workshop); parts of this paper came together during discussions at these venues. 
T.S. also thanks the Oberwolfach Research Institute for Mathematics (and the organizers of the Proof Complexity and Beyond workshop), the Simons Institute (and the organizers of the Optimization semester program), and the Banff International Research Station (and the organizers of the Approximation Algorithms and Hardness workshop) where she tried to prove the opposite of the results in this paper, as well as Sam Hopkins, with whom some of those efforts were made.

\bibliographystyle{alpha}
\bibliography{sa}

\appendix
\section{Characteristics of distributions of XOR-subformula coefficients}\label{app:Ws}
We now prove \Cref{lem:w-props}.
We will use the following estimate for the mean absolte deviation of a binomial random variable.
\begin{lemma}[e.g. \cite{blyth1980}]\label{lem:bin-mad}
If $X$ is distributed according to the binomial distribution $X \sim \Bin(n,p)$, then $\E |X - \E X| = \sqrt{\frac{2}{\pi}np(1-p)} + O(\frac{1}{\sqrt n})$.
\end{lemma}

\begin{proof}[Proof of \cref{lem:w-props}]
We calculate the absolute value directly.
Given that there are exactly $k$ nonzero $Y_t$, the absolute value of $X$ is distributed according to $|\Bin(k,\tfrac{1}{2}) - \tfrac{1}{2}k|$. 
Using the method of conditional expectations,
\begin{align*}
\E |X|
= \sum_{k=0}^{N} \Pr \Ind[k \text{ nonzero } Y_t\text{'s}] \cdot \E |\Bin(k,\tfrac{t}{2}) - \tfrac{1}{2}k|
&\ge \sum_{k=0}^N \binom{N}{k} p^k (1-p)^{N-k} \cdot \sqrt{\frac{1}{2\pi}k},
\end{align*}
where we have applied the estimate from \Cref{lem:bin-mad}.
Letting $D(a\| b) = a\ln\frac{a}{b} + (1-a)\ln\frac{1-a}{1-b}$ be the relative entropy, we then have from Stirling's inequality that
\begin{align}
\E|X| 
\ge \sum_{k=1}^N \sqrt{\frac{2\pi}{e^2}\frac{N}{k(N-k)}} \cdot \exp\left(-N\cdot D\left(\tfrac{k}{N}\|p\right)\right)\cdot \sqrt{\frac{1}{2\pi}k}
&\ge\frac{1}{e}\sum_{k=1}^N \exp\left(-N\cdot D\left(\tfrac{k}{N}\|p\right)\right),\label{eq:both-bd}
\end{align}
Now, if $pN < 1$, we take
\begin{align}
\cref{eq:both-bd}
&\ge \frac{1}{e}\sum_{k=1}^N \exp\left(k\log\left(\frac{pN}{k}\right)\right)
= \frac{1}{e}\sum_{k=1}^N \left(\frac{pN}{k}\right)^k
\ge -\frac{1}{e}\log(1-pN) - O(p^N)
\end{align}
as desired.

If $pN \ge 1$, applying the change of variables $\ell = k - \lfloor pN\rfloor$ and $\delta = \frac{\ell}{N}$,
\begin{align}
\cref{eq:both-bd}
&\ge\frac{1}{e}\sum_{\substack{\ell = 1- \lfloor pN\rfloor\\\delta = \frac{\ell}{N}}}^{\lfloor(1-p)N\rfloor} \exp\left(-N\cdot D\left(p + \delta\|p\right)\right)\label{eq:var-chng}
\end{align}
Now using the Taylor expansion for $\log(1-x)$ to simplify and restricting the sum over the range $\ell = [-\lfloor \sqrt{pN}\rfloor,\lfloor\sqrt{pN}\rfloor]$, we get the bound
\begin{align}
\cref{eq:var-chng}
&\ge\frac{1}{e}\sum_{\substack{\ell = -\lfloor\sqrt{pN}\rfloor\\\delta = \frac{\ell}{N}}}^{\lfloor\sqrt{pN}\rfloor} \exp\left(-N\frac{\delta^2}{2}\right)
\ge\frac{1}{e}\cdot 2\sqrt{pN} \cdot \exp(-\tfrac{p}{2})
\ge\frac{2}{e^{3/2}}\sqrt{pN},
\end{align}
as desired.

The first and second moment we can also obtain by calculation; the $Z_t$ ensure that the summands have mean $0$, and the $Y_t$ give that the variance of the summands is $p$, which gives the result.

The tail bound $\Pr(|X| \ge (1+2t)\sqrt{pN}) \le 2\exp(-t^2)$ comes from an application of Bernstein's inequality if $pN \ge 1$; when $pN < 1$, we again apply Bernstein's inequality, in which case we have
\[
\Pr\left(|X| - \E|X| \ge s \right) \le \exp\left(-\frac{1}{2}\frac{s^2}{ pN + \frac{1}{3} s}\right),
\]
and choosing $s = 1 + t$ gives the result.
\end{proof}

\end{document}